\documentclass[11pt,DIV=12,a4paper]{scrartcl}

\usepackage{amsmath,amsfonts,amsthm}
\usepackage{cite,booktabs,tabularx}
\usepackage{tikz,authblk}
\usepackage{enumitem}

\setkomafont{title}{\normalfont \LARGE \boldmath}
\setkomafont{section}{\normalfont \Large \bfseries \boldmath \boldmath}
\setkomafont{subsection}{\normalfont \large \bfseries \boldmath}
\setkomafont{subsubsection}{\normalfont \bfseries \boldmath}
\setkomafont{paragraph}{\normalfont \bfseries \boldmath}
\setkomafont{descriptionlabel}{\normalfont \itshape}

\newtheorem{claim}{Claim}[section]
\newtheorem{theorem}[claim]{Theorem}
\newtheorem{lemma}[claim]{Lemma}

\newtheorem{observation}[claim]{Observation}

\usetikzlibrary{decorations.pathmorphing}
\usetikzlibrary{decorations.pathreplacing}
\usetikzlibrary{calc}

\newcommand{\edge}[2]{\ensuremath{\{X_{#1}, X_{#2}\}}}
\newcommand{\real}{\ensuremath{\mathbb{R}}}
\newcommand{\nat}{\ensuremath{\mathbb{N}}}
\newcommand{\eps}{\varepsilon}
\newcommand{\maxx}{D_{\max}}
\newcommand{\dmin}{\Delta_{\min}}
\newcommand{\dminl}{\Delta_{\min}^{\operatorname{link}}}

\DeclareMathOperator{\probab}{\mathbb{P}}
\DeclareMathOperator{\expected}{\mathbb{E}}

\tikzstyle{Point}=[circle, draw, fill=black, inner sep=0mm, minimum size=1.5mm]
\tikzstyle{Line}=[line width=0.6pt]
\tikzstyle{Brace}=[Line, decorate, decoration={brace, amplitude=1mm}]

\title{Smoothed Analysis of the 2-Opt Heuristic for the TSP under Gaussian Noise \thanks{This paper is based on results
presented at ISAAC 2013~\cite{MantheyVeenstra:2Opt:2013} and ICALP 2015~\cite{KuennemannManthey:2Opt:2015}.}}
\author[1]{Marvin K\"unnemann}
\author[2]{Bodo Manthey}
\author[2]{Rianne Veenstra}

\affil[1]{Max Planck Institute for Informatics, Saarbr\"ucken, Germany, and Saarbr\"ucken Graduate School of Computer Science, \texttt{marvin@mpi-inf.mpg.de}}
\affil[2]{University of Twente, Department of Applied Mathematics, Enschede, Netherlands, \texttt{b.manthey@utwente.nl}}

\usepackage{nicefrac}
\newcommand{\one}{\mathbf{1}}
\newcommand{\vol}{\mathrm{vol}}
\newcommand{\pos}{\mathrm{pos}}
\DeclareMathOperator{\argmin}{argmin}
\DeclareMathOperator{\argmax}{argmax}
\newcommand{\bTSP}{\mathsf{TSP}_B}
\DeclareMathOperator{\polylog}{polylog}
\newcommand{\distance}{\mathrm{dist}}
\newcommand{\Ttwo}{{\cal T}_2}
\newcommand{\Tone}{{\cal T}_1}
\newcommand{\pad}{\mathrm{pad}}
\newcommand{\lemref}[1]{Lemma~\ref{lem:#1}}
\newcommand{\thmref}[1]{Theorem~\ref{thm:#1}}
\newcommand{\obsref}[1]{Observation~\ref{obs:#1}}
\newcommand{\itemref}[1]{\ref{itm:#1}}
\newcommand{\norm}[1]{\left\Vert #1 \right\Vert}
\newcommand{\TSP}{\mathsf{TSP}}
\newcommand{\TwoOPT}{\mathsf{2OPT}}
\DeclareMathOperator{\Gauss}{\mathcal{N}}
\DeclareMathOperator{\diag}{diag}
\DeclareMathOperator{\pert}{pert}
\newcommand{\orig}{\mathrm{orig}}
\newcommand{\MST}{\mathsf{MST}}
\newcommand{\OPT}{\mathrm{OPT}}

\begin{document}

\maketitle

\thispagestyle{plain}
\begin{abstract}
The 2-opt heuristic is a very simple local search heuristic for the
traveling salesperson problem. In practice it usually converges quickly
to solutions within a few percentages of optimality. In contrast to this, its running-time
is exponential and its approximation performance is poor in the worst case.

Englert, R\"oglin, and V\"ocking (\emph{Algorithmica}, 2014)
provided a smoothed analysis in the so-called
one-step model in order to explain the performance of 2-opt on $d$-dimensional
Euclidean instances, both in terms of running-time and in terms of approximation ratio. However, translating their results to the classical
model of smoothed analysis, where points are perturbed by Gaussian distributions with standard deviation $\sigma$, yields only
weak bounds.

We prove bounds that are polynomial in $n$ and $1/\sigma$
for the smoothed running-time with Gaussian perturbations.
In addition, our analysis for Euclidean distances is much simpler than
the existing smoothed analysis.

Furthermore, we prove a smoothed approximation ratio of $O(\log (1/\sigma))$. This bound is almost tight,
as we also provide a lower bound of $\Omega(\frac{\log n}{\log\log n})$ for $\sigma = O(1/\sqrt n)$.
Our main technical novelty here is that, different from existing smoothed analyses, we do not separately analyze objective values of the global and local 
optimum on all inputs (which only allows for a bound of $O(1/\sigma)$), but simultaneously bound them on the same input.
\end{abstract}

%\todo{apx part: is OPT and TSP the same?}
%
%\todo{do we need ``dist''?}

\section{2-Opt and Smoothed Analysis}
\label{sec:intro}

The traveling salesperson problem (TSP) is one of the classical combinatorial optimization problems.
Euclidean TSP is the following variant: given points $X \subseteq [0,1]^d$, find the shortest Hamiltonian
cycle that visits all points in $X$ (also called a \emph{tour}). Even this restricted variant is NP-hard for $d \geq 2$~\cite{Papadimitriou:EuclideanTSP:1977}.
We consider Euclidean TSP with Manhattan and Euclidean distances as well as squared
Euclidean distances to measure the distances between points.
For the former two, there exist polynomial-time approximation schemes (PTAS)~\cite{Arora:PTASEuclideanTSP:1998,Mitchell:EuclideanPTAS:1999}. The latter,
which has applications in power assignment problems for wireless networks~\cite{FunkeEA:PowerTSP:2011}, admits a
PTAS for $d=2$ and is APX-hard for $d \geq 3$~\cite{NijnattenEA:TSPSquared:2010}.

As it is unlikely that there are efficient algorithms for solving Euclidean TSP optimally,
heuristics have been developed in order to find near-optimal solutions quickly. One very simple and popular heuristic
is 2-opt: starting from an initial tour, we iteratively replace two edges by two other edges to obtain a shorter
tour until we have found a local optimum.
Experiments indicate that 2-opt converges to near-optimal solutions quite
quickly~\cite{JohnsonMcGeoch:CaseTSP:1997,JohnsonMcGeoch:ExperimentalSTSP:2002},
but its worst-case performance is bad: the worst-case running-time is
exponential even for $d=2$~\cite{EnglertEA:2Opt:2014}
and the approximation ratio is
$\Omega(\log n/\log \log n)$ for Euclidean instances~\cite{ChandraEA:OldOpt:1999}.

An alternative to worst-case analysis is average-case analysis,
where the expected performance with respect to some probability distribution is measured.
The average-case running-time for Euclidean and random metric instances and the average-case approximation
ratio for non-metric instances of 2-opt have been analyzed~\cite{Kern:Switching:1989,EngelsManthey:2Opt:2009,ChandraEA:OldOpt:1999,BringmannEA:RSP:2013}.
However, while worst-case analysis is often too pessimistic because it is dominated by artificial instances
that are rarely encountered in practice, average-case analysis is dominated by random instances, which have
often very special properties with high probability that they do not share with typical instances.

In order to overcome the drawbacks of both worst-case and average-case analysis
and to explain the performance
of the simplex method, Spielman and Teng invented smoothed analysis~\cite{SpielmanTeng:SmoothedAnalysisWhy:2004},
a hybrid of worst-case and average-case analysis:
an adversary specifies an instance, and then this instance is slightly randomly perturbed. The smoothed performance
is the expected performance, where the expected value is taken over the random perturbation.
The underlying assumption is that real-world instances
are often subjected to a small amount of random noise. This noise can stem from measurement or rounding errors,
or it might be a realistic assumption that the instances are influenced by unknown circumstances, but we do not have
any reason to believe that these are adversarial.
Smoothed analysis often allows more realistic conclusions about the
performance than worst-case or average-case analysis.
Since its invention, it has been applied successfully
to explain the performance of a variety of
algorithms.
%~\cite{ArthurEA:kMeansPoly:2011,BrunschEA:SSP:2013,MantheyRoeglin:Bregman:2013,ArthurVassilvitskii:ICP:2009,Vershynin:SmoothedHirsch:2009,DamerowEA:LTRM:2012}.
We refer to two surveys for an overview of smoothed analysis in general~\cite{SpielmanTeng:CACM:2009,MantheyRoeglin:SmoothedSurvey:2011} and a more recent survey about
smoothed analysis applied to local
search algorithms~\cite{Manthey:ABWCBook}.

\subsection{Related Results}

\paragraph{Running-time.}
Englert, R\"oglin, and V\"ocking~\cite{EnglertEA:2Opt:2014} provided a smoothed analysis of 2-opt in order to explain its performance.
They used the \emph{one-step model}:
an adversary specifies $n$ probability density functions $f_1, \ldots, f_n: [0,1]^d \to [0,\phi]$. Then the $n$ points $X_1, \ldots, X_n$
are drawn independently according to the densities $f_1, \ldots, f_n$, respectively. Here, $\phi$ is the perturbation parameter.
If $\phi =1$, then the only possibility is the uniform distribution on $[0,1]^d$, and we obtain an
average-case analysis.
The larger $\phi$, the more powerful the adversary.
Englert et al.~\cite{EnglertEA:2Opt:2014} proved that the expected number of iterations of 2-opt is
$O(n^{4}  \phi)$ and
$O(n^{4+\frac 13}  \phi^{\frac 83}  \log(n\phi))$ for Manhattan and
Euclidean distances, respectively.
These bounds can be improved slightly by choosing the initial tour with an insertion heuristic.
However, if we transfer these bounds to the classical model of
points perturbed by Gaussian distributions of standard deviation $\sigma$,
we obtain bounds that are polynomial in $n$ and $1/\sigma^d$~\cite[Section 6]{EnglertEA:2Opt:2014},
since the maximum density of a $d$-dimensional Gaussian with standard deviation
$\sigma$ is $\Theta(\sigma^{-d})$. While this is polynomial for any fixed $d$, it
is unsatisfactory that the degree of the polynomial depends on $d$.

\paragraph{Approximation ratio.}

Much less is known about the smoothed approximation performance of algorithms.
Karger and Onak have shown that multi-dimensional bin packing can be approximated arbitrarily well for smoothed instances~\cite{KargerOnak:SmoothedPacking:2007} and there are frameworks to approximate Euclidean optimization problems such as TSP for smoothed
instances~\cite{BlaeserEA:Partitioning:2013,CurticapeanK15}. However, these approaches mostly consider algorithms tailored to solving smoothed instances. 

With respect to concrete algorithms other than 2-opt, we are only aware of analyses of the jump and lex-jump heuristics for scheduling~\cite{BrunschEA:Scheduling:2014,Etscheid15}.

Englert et al.~\cite{EnglertEA:2Opt:2014} proved a bound of $O(\phi^{1/d})$. Translated to Gaussians,
this yields a bound of $O(1/\sigma)$ if we truncate the Gaussians such that all points
lie in a hypercube of constant side length.
This result, however, does not explain the approximation performance 2-opt, as the
bound is still quite large, even for larger values of $\sigma$ or smaller values of $\phi$.

\subsection{Our Contribution}

In order to improve our understanding of the practical performance of 2-Opt, we provide an improved smoothed analysis
of both its running-time and its approximation ratio. To do this, we use the classical smoothed
analysis model: an adversary chooses $n$ points from the $d$-dimensional unit hypercube $[0,1]^d$, and then these points are independently
randomly perturbed by Gaussian random variables of standard deviation~$\sigma$.

\paragraph{Running-time.}

The bounds that we prove are polynomial in $n$ and $1/\sigma$.
Different to earlier results, the degree of the polynomial is independent of $d$. As distance measures, we
consider Manhattan (Section~\ref{sec:manhattan}), Euclidean
(Section~\ref{sec:euc}), and
squared Euclidean distances (Section~\ref{sec:sed}).

The analysis for Manhattan distances is essentially an adaptation of the existing analysis
by Englert et al.~\cite[Section 4.1]{EnglertEA:2Opt:2014}. Note that our bound does not have any factor that is exponential in $d$.

Our analysis for Euclidean distances is considerably
simpler than the one by Englert et al., which is rather technical and takes more than 25
pages~\cite[Section~4.2 and Appendix~C]{EnglertEA:2Opt:2014}.

The analysis for squared Euclidean distances is, to our knowledge, not preceded by a smoothed analysis in the one-step
model. Because of the nice properties of squared Euclidean distances and Gaussian perturbations,
this smoothed analysis is relatively compact and elegant.

Table~\ref{tab:results} summarizes our bounds for the number of iterations.

\begin{table}[t]
\centering
\begin{tabularx}{\textwidth}{@{}lllX@{}} \toprule
& Manhattan & Euclidean & squared Euclidean \\ \midrule
Englert et al.~\cite{EnglertEA:2Opt:2014} & $ 2^{O(d)} n^4 \phi$ & $O_d\bigl(n^{4 + \frac 13} \phi^{\frac 83} \log(n\phi)\bigr)$ & -- \\ \midrule
general & $O\bigl(\frac{d^2 n^4 \maxx}{\sigma}\bigr)$ 
& $O\bigl(\frac{\sqrt d n^4 \maxx^4}{\sigma^4}\bigr)$& $O\bigl(\frac{\sqrt d n^4\maxx^2}{\sigma^2}\bigr)$  \\[\smallskipamount]
$\sigma = O(1/\sqrt{n \log n})$ &$O\bigl(\frac{d^2 n^4}{\sigma}\bigr)$ 
& $O\bigl(\frac{\sqrt d n^4}{\sigma^4}\bigr)$ & $O\bigl(\frac{\sqrt d n^4}{\sigma^2}\bigr)$ \\[\smallskipamount]
$\sigma = \Omega(1/\sqrt{n \log n})$ &$O\bigl(d^2 n^5 \sqrt{\log n}\bigr)$ 
& $O\bigl(\sqrt d n^6 \log^2 n\bigr)$ & $O\bigl(\sqrt d n^5 \log n\bigr)$ \\
remarks & & only for $d \geq 4$ & only for $d \geq 3$; a weaker bound holds for $d=2$ (Theorem~\ref{thm:squaredsingle}) \\ \bottomrule
\end{tabularx}
\caption{Our bounds compared to the bounds obtained by Englert et al.~\cite{EnglertEA:2Opt:2014} for the one-step model.
The bounds can roughly be transferred to Gaussian noise by replacing $\phi$ with $\sigma^{-d}$.
For convenience, we added our bounds for small and large values of $\sigma$: for $\sigma = O(1/\sqrt{n \log n})$,
we have $\maxx = \Theta(1)$, for larger $\sigma$, we have $\maxx = \Theta(\sigma \sqrt{n \log n})$. The notation $O_d$ means that
terms depending on $d$ are hidden in the $O$. The remarks are only for our bounds.}
\label{tab:results}
\end{table}

\paragraph{Approximation ratio.}

As the earlier smoothed analysis by Englert et al.~\cite{EnglertEA:2Opt:2014}, we provide bounds on the quality of the worst local optimum.
While this measure is rather unrealistic and pessimistic, it decouples the analysis from the seeding of the heuristic. Taking
into account the seeding would probably severely complicate the analysis.

Our bound of $O(\log(1/\sigma))$ improves significantly upon the direct
translation of the bound of Englert et al.~\cite{EnglertEA:2Opt:2014} to Gaussian perturbations (see Section~\ref{sec:length} for how to
translate the bound to Gaussian perturbations without truncation).
It smoothly interpolates between the average-case constant approximation ratio and the worst-case bound of $O(\log n)$. 

In order to obtain our improved bound for the smoothed approximation ratio, we take into account the origins 
of the points, i.e., their unperturbed positions. Although this information is not available to the algorithm,
it can be exploited in the analysis.
The smoothed analyses of approximation ratios so far~\cite{EnglertEA:2Opt:2014,BrunschEA:Scheduling:2014,CurticapeanK15,BlaeserEA:Partitioning:2013,Etscheid15,KargerOnak:SmoothedPacking:2007}
essentially ignored this information. While this simplifies the analysis,
being oblivious to the unperturbed positions seems to be too pessimistic.
In fact, we see that the bound of Englert et al.~\cite{EnglertEA:2Opt:2014} cannot be improved beyond $O(1/\sigma)$
by ignoring the positions of the points (Section~\ref{sec:length}).
The reason for this limitation is that the lower bound for the global optimum is obtained
if all points have the same origin, which corresponds to an average-case rather than a smoothed analysis.
On the other hand, the upper bound for the local optimum has to hold for all choices of the unperturbed points, most of which
yield higher costs for the global optimum than the average-case analysis. Taking this into account carefully yields our bound of $O(\log(1/\sigma))$
(Section~\ref{sec:upper}).

To complement our upper bound, we show that the lower bound of $\Omega(\log n/\log \log n)$ by Chandra et al.~\cite{ChandraEA:OldOpt:1999} remains
true for $\sigma = O(1/\sqrt n)$ (Section~\ref{sec:lb}). This implies that a smoothed bound
of $o(\log(1/\sigma)/\log\log(1/\sigma))$ is impossible, and, thus, our bound cannot be improved significantly.

%\subsection{Outline}
%
%In the next section, we introduce 2-opt and our smoothed model.
%The main technical sections are Section~\ref{sec:rt} and~\ref{sec:apx}. Both start with technical preparation that is needed for the proofs.
%With respect to the running-time, we first consider Manhattan distances (Section~\ref{sec:manhattan}), then squared Euclidean distances (Section~\ref{sec:sed}),
%and finally Euclidean distances (Section~\ref{sec:euc}), as Euclidean distances are the most difficult to handle.
%
%After that, we analyze 2-opt's approximation performance. We only consider the approximation ratio
%with Euclidean distances. We start by transferring the result by Englert et al.~\cite{EnglertEA:2Opt:2014}
%to Gaussian perturbations without truncation (Section~\ref{sec:length}). We also observe that no better bound than $O(1/\sigma)$ is possible
%when considering the global and local optima separately.
%After that, we improve this bound to $O(\log(1/\sigma))$ (Section~\ref{sec:upper}).
%Finally, we show that this bound cannot be significantly improved (Section~\ref{sec:lb}).
%

\section{2-Opt and Smoothing Model}

\subsection{2-Opt Heuristic for the TSP}

Let $X = \subseteq \real^d$ be a set of $n$ points.
The goal of the TSP is to find a Hamiltonian cycle (also called a tour)
$T$ through $X$ that has minimum length according to some distance measure.
In this paper, we consider standard Euclidean distances for both approximation ratio and running-time
as well as
squared Euclidean distances and Manhattan distances for the running-time.

Given a tour $T$, a \emph{2-change}
replaces two edges \edge 12\ and \edge 34\ of $T$ by two new
edges \edge 13\ and \edge 24,
provided that this yields again a tour
(this is the case if $X_1, X_2, X_3, X_4$ appear in this order in the tour)
and that this decreases the length of the tour, i.e.,
$d(X_1, X_2) + d(X_3, X_4) - d(X_1, X_3) - d(X_2, X_4) > 0$,
where $d(a,b) = \|a-b\|_2$ (Euclidean distances),
$d(a,b) = \|a-b\|_1$ (Manhattan distances), or
$d(a,b) = \|a-b\|_2^2$ (squared Euclidean distances). The 2-opt heuristic
iteratively improves an initial tour by applying 2-changes until it reaches a local optimum.
A local optimum is called a \emph{2-optimal} tour.

\subsection{Smoothing Model}

Throughout the rest of this paper, let $\bar X = \{x_1, \ldots, x_n\} \in [0,1]^d$ be a set of $n$ points from the unit hypercube.
In the smoothed analysis, these points are chosen by an adversary, and they serve as unperturbed \emph{origins}.
Let $Z_1, \ldots, Z_n \sim \mathcal N (0, \sigma^2)$ be $n$ independent random variables with mean $0$ and standard deviation $\sigma$.
By slight abuse of notation, $\mathcal N$ refers here to the multivariate normal distribution with covariance
matrix $\diag(\sigma^2)$.
We obtain the \emph{perturbed} point set $X = \{X_1, \ldots, X_n\} \subseteq \real^d$ by adding $X_i = x_i + Z_i$ for each $i \in [n] = \{1, \ldots, n\}$.
We write $X \leftarrow \pert_\sigma(\bar X)$ to make explicit from which
point set $\bar X$ the points in $X$ are obtained.

We assume that $\sigma \leq 1$ throughout the paper. This is justified by two reasons.
First, small $\sigma$ are the interesting case, i.e.,
when the order of magnitude of the perturbation is relatively small.
Second,
smoothed performance guarantees are monotonically decreasing
in $\sigma$: if we have $\sigma > 1$, then this is equivalent
to adversarial instances in $[0, 1/\sigma]^d$ that are perturbed with standard deviation $1$.
This in turn is dominated by adversarial instances in $[0,1]^d$ that are perturbed with standard deviation $1$,
as $[0,1/\sigma]^d \subseteq [0,1]^d$. Thus, any upper bound for $\sigma = 1$ (be it for the number of iterations or the approximation ratio)
holds also for larger $\sigma$.

%In contrast to the model that we just described, let us introduce also the model used by Englert et al.~\cite{EnglertEA:2Opt:2014}, as we refer to
%their results a few times.
%The $\phi$-bounded perturbation model (also called \emph{one-step model}) lets the adversary directly specify (not necessarily identical) distributions by choosing probability density functions $f_1,\dots,f_n:[0,1]^d\to [0,\phi]$. The perturbed input is then generated by independently sampling $X_1\sim f_1,\dots,X_n\sim f_n$. Note that the resulting input is always contained in $[0,1]^d$ and with higher $\phi$, the adversary can concentrate points to smaller regions of the input space. Roughly speaking, when translating Gaussian perturbations to the one-step model, $\phi$ is proportional to $\sigma^{-d}$ for fixed $d$.

Let us make a final remark about the smoothing model: while the algorithm itself, the 2-opt heuristic in our case, only sees $X$ and
does not know anything about the origins $\bar X$, we can of course exploit the positions of the unperturbed points in the analysis.

\section{Smoothed Analysis of the Running-time}
\label{sec:rt}

In this section, we make the dependence on all parameters (the number $n$ of points, the dimension~$d$, and the perturbation parameter $\sigma$) explicit.
This means that the $O$ or $\Omega$ do not hide any factors, not even factors depending  on $d$, which is often considered as a constant and therefore ignored. (This is also in contrast to our analysis of the approximation ratio, where the hidden constant can indeed depend on $d$.)

\subsection{Probability Theory for the Running-time}

In order to get an upper bound for the length of the initial tour, we need an upper bound
for the diameter of the point set $X$. Such an upper bound is also necessary for the analysis of 2-changes
with Euclidean distances (Section~\ref{sec:euc}).
We choose $\maxx$ such that $X \subseteq [-\maxx, \maxx]^d$ with a probability of at least
$1-1/n!$. For fixed $d$ and $\sigma \leq 1$, we can choose $\maxx = \Theta(1+\sigma
\sqrt{n \log n})$ according to the following lemma.
For $\sigma = O(1/\sqrt{n \log n})$, we have $\maxx = \Theta(1)$.

\begin{lemma}
\label{lem:maxx}
Let $c\geq 2$ be a sufficiently large constant, and let $\maxx = c \cdot
(\sigma\sqrt{n \log n} + 1)$. Then $\probab(X \not\subseteq [-\maxx, \maxx]^d) \leq 1/n!$.
\end{lemma}

\begin{proof}
We have $X \not\subseteq [-\maxx, \maxx]^d$ only if there is a point $x_i$ and a coordinate
of $x_i$ that is perturbed by more than $\maxx-1 \geq c \sigma \cdot \sqrt{n \log n}$. According to
Durrett~\cite[Theorem 1.2.3]{Durrett:Probability:2013},
the probability that a 1-dimensional Gaussian of standard deviation $\sigma$ is more than $c \sigma \sqrt{n \log n}$ away
from its mean
is bounded from above by $2 \cdot \frac{\exp(-c^2 n \log n /2)}{c \sqrt{2\pi n \log n}}$.
Thus, the probability that $X \not\subseteq [-\maxx, \maxx]^d$ is bounded from above
by $2dn \cdot \frac{\exp(-c^2 n \log n /2)}{c \sqrt{2\pi n \log n}}$.
For sufficiently large $c$, this is at most $1/n!$.
\end{proof}

Note that the constant $c$ in Lemma~\ref{lem:maxx} does not depend on the dimension $d$.

The following lemma is well known and follows from the
fact that the density of a $d$-dimensional Gaussian with standard deviation
$\sigma$ is bounded from above by $(2\sigma)^{-d}$ and the volume of a $d$-dimensional ball of radius
$\eps$ is bounded from above by $(2\eps)^d$.

\begin{lemma}
\label{lem:ball}
Let $a \in \real^d$ be drawn according to a $d$-dimensional Gaussian distribution
of standard deviation $\sigma$, and let $B = \{b \in \real^d \mid \|b-c\|_2 \leq \eps\}$ be a $d$-dimensional hyperball
of radius $\eps$ centered at $c \in \real^d$.
Then $\probab(a \in B) \leq (\eps/\sigma)^d$.
\end{lemma}

For $x, y \in \real^d$ with $x \neq y$, let
$L(x,y) = \{\xi \cdot (y-x)  + x \mid \xi \in \real\}$
denote the straight line through $x$ and $y$.

\begin{lemma}
\label{lem:closeline}
Let $a, b  \in \real^d$ be arbitrary with $a \neq b$. Let $c \in \real^d$
be drawn according to a $d$-dimensional Gaussian distribution with standard deviation $\sigma$.
Then the probability that $c$ is $\eps$-close to $L(a, b)$, i.e.,
$\min_{c^\star \in L(a,b)} \|c-c^\star\|_2 \leq \eps$, is bounded from above by
$(\eps/\sigma)^{d-1}$.
\end{lemma}

\begin{proof}
We divide drawing $c$ into drawing a $1$-dimensional Gaussian $c^\star$ in
the direction of $a-b$
and drawing a $(d-1)$-dimensional Gaussian $c'$ in the hyperplane orthogonal to $a-b$
and containing $c^\star$. Then the distance of $c$ to $L(a,b)$ is $\|c-c^\star\|_2$.
For every $c^\star$, the point $c$ is $\eps$-close to $L(a,b)$ only if
$c'$ falls into a $(d-1)$-dimensional hyperball of radius $\eps$ around $c^\star$
in the $(d-1)$-dimensional subspace orthogonal to $a-b$.
Now the lemma follows by applying Lemma~\ref{lem:ball}.
\end{proof}

We need the following lemma in Section~\ref{sec:euc}.

\begin{lemma}
\label{lem:fctgauss}
Let $f: \real \to \real$ be a differentiable function.
Let $B$ be an upper bound for the absolute value of the derivative of $f$. Let $c$ be distributed
according to a Gaussian distribution with standard deviation $\sigma$.
Let $I$ be an interval of size $\eps$, and let $f(I) = \{f(x) \mid x \in I\}$
be the image of $I$.
Then $\probab(c \in f(I)) = O(B\eps/\sigma)$.
\end{lemma}

\begin{proof}
Since the derivative of $f$ is bounded by $B$, the set $f(I)$ is contained
in some interval of length $B\eps$. The lemma follows
since the density of $c$ is bounded from above by $O(1/\sigma)$.
\end{proof}

The chi distribution~\cite[Section 8]{EvansEA:StatisticalDist:2000} is the distribution of the Euclidean length of a $d$-dimensional Gaussian 
random vector of standard deviation $\sigma$ and mean $0$. In the following,
we denote its density function by $\chi_d$. It is given
by
\begin{equation}
\chi_{d,\sigma}(x) = \frac{2^{1-\frac d2} \cdot \left(\frac x\sigma\right)^{d-1} \cdot \exp\bigl(-(x/\sigma)^2/2\bigr)}{\sigma \cdot \Gamma(d/2)}, \label{equ:chi}
\end{equation}
where $\Gamma$ denotes the gamma function.
We need the following lemma several times.

\begin{lemma}
\label{lem:integral}
Assume that $c \in \nat$ is a fixed constant and $d \in \nat$ is arbitrary with $d > c$. Then we have
\[
\int_0^\infty  \chi_{d,\sigma}(x) x^{-c} \, \textup dx = 
\frac{2^{-c/2} \Gamma\left(\frac{d-c}2\right)}{\sigma^c \cdot \Gamma\left(\frac d2 \right)} = 
 \Theta\left(\frac{1}{d^{c/2} \cdot \sigma^c}\right).
\]
\end{lemma}

\begin{proof}
The first equality follows by integration. For the second inequality, we observe $2^{-c/2}$ is a fixed constant (which also never depends on $d$
when we apply the lemma) and that
\[
  \Gamma(x) = \sqrt{2\pi} x^{x-1/2} e^{-x + \mu(x)}
\]
for some function $\mu$ with $\mu(x) \in \bigl[0, \frac 1{12x}\bigr]$ according to Stirling's formula~\cite[6.1.37]{Abramowitz:Pocket:1984}.
We have $\frac{d-c}2 \geq \frac 12$ as $d > c$ and both are integers.
Then
\begin{align*}
\frac{\Gamma(\frac{d-c}2)}{\Gamma(\frac d2)}
& = \frac{\sqrt{2\pi} \cdot \left(\frac{d-c}2\right)^{\frac{d-c-1}{2}} \cdot 
   \exp\left(-\frac{d-c}2 + \mu\left(\frac{d-c}2\right)\right)}{\sqrt{2\pi}
      \cdot \left(\frac d2\right)^{\frac{d-1}2} \cdot \exp\left(-\frac d2 + \mu\left(\frac d2\right)\right)} \\
& = \frac{\left(\frac{d-c}2\right)^{\frac{d-c-1}{2}}}{\left(\frac d2\right)^{\frac{d-1}2}} 
\cdot \underbrace{\exp\left(\frac{c}2 + \mu\left(\frac{d-c}2\right) - \mu\left(\frac d2\right)\right)}_{=\Theta(1)} \\
& = \underbrace{\left(\frac{d-c}d\right)^{\frac{d-1}2}}_{=A} \cdot \underbrace{\left(\frac{d-c}2\right)^{-\frac c2}}_{=B}
\cdot \Theta(1).
\end{align*}
Here, the third equality follows from two facts: first, $c$ is a fixed constant, thus $\exp(c/2) = \Theta(1)$. Second, $\frac{d-c}2, \frac d2 \geq \frac 12$.
Thus, $\mu(\frac{d-c}2)$ and $\mu(\frac d2)$ lie between $0$ and a constant. Hence, the exponential term is $\Theta(1)$.

Analyzing $A$ and $B$ remains to be done: We have $B \leq (d/2)^{-c/2}$, thus $B = O(d^{-c/2})$.
If $d \leq 2c$, then $B$ is bounded from below by a constant and so is $d^{-c/2}$. If $d \geq 2c$,
then $B \geq (d/4)^{-c/2} = \Omega(d^{-c/2})$. Hence, $B = \Theta(d^{-c/2})$.

We have $A = \bigl(1 - \frac cd\bigr)^{\frac{d-1}2} \leq \exp\bigl(-\frac{(d-1)\cdot c}{2d}\bigr) = O(1)$.
Distinguishing the cases $d \leq 2c$ and $d > 2c$ in the same way as for $B$ yields $A = \Omega(1)$. Thus, $A = \Theta(1)$.
\end{proof}

The analysis with Euclidean and squared Euclidean distances depends on the distribution of the distance between two points perturbed by Gaussians,
where a larger distance between the two points is better for the analysis.
The following two lemmas show that, given that larger distance is better, we can replace the distribution of the distance by the corresponding
chi distribution. Since we do not know the original positions of the points involved, this allows us to replace unknown distributions by the
chi distribution.

\begin{lemma}
\label{lem:gaussdom}
Assume that $a$ is drawn according to a $d$-dimensional Gaussian distribution with standard deviation $\sigma$ and mean $0$.
Assume that $b$ is drawn according to a $d$-dimensional Gaussian distribution with standard deviation $\sigma$ and mean $\mu$.
Then $\|b\|_2$ stochastically dominates $\|a\|_2$, i.e.,
$\probab(\|b\|_2 \leq t) \leq \probab(\|a\|_2 \leq t)$ for all $t \in \real$.
\end{lemma}

\begin{proof}
For $d=1$, we have the following:
\begin{align*}
 \probab\bigl(\|b\|_2 \leq t\bigr) & = \probab\bigl(b \in [-t,t]\bigr) 
 = \probab\bigl(a \in [-t-\mu, t-\mu]\bigr) \\
 & = \probab\bigl(a \in [-t, t]) + \underbrace{\probab\bigl(a \in [-t-\mu, -t]\bigr) - \probab\bigl(a \in [t-\mu, t]\bigr)}_{\leq 0}
  \leq \probab\bigl(\|a \|_2 \leq t\bigr)
\end{align*}
Now we prove the lemma for larger $d$. Since Gaussian distributions are rotation symmetric,
we can assume that $\mu = (\delta, 0, \ldots, 0)$ for some $\delta \geq 0$.

We observe that $\|b\|_2$ dominates $\|a\|_2$ if and only if
$\|b\|_2^2$ dominates $\|a\|_2^2$. Let $b' = a + \mu$. It suffices to prove the lemma for this choice of $b'$, as $b'$ follows the same distribution as $b$.
Fixing $a_2, \ldots, a_d$ fixes also $a_2', \ldots, a_d'$.
Then $\|b'\|_2^2$ dominates $\|a\|_2^2$ if $|a_1 + \delta|$ dominates $|a_1|$. This is true because
the lemma holds for $d=1$.
\end{proof}

\begin{lemma}
\label{lem:replacechi}
Let $b$ be as in Lemma~\ref{lem:gaussdom}, and let $h: [0, \infty] \to [0,\infty)$ be a monotonically decreasing function.
Let $g$ be the density function of $\|b\|$.
Then
\[
 \int_{0}^\infty g(x) h(x) \, \textup d x \leq \int_{0}^\infty \chi_{d, \sigma}(x) h(x) \, \textup d x,
\]
provided that both integrals exist.
\end{lemma}

\begin{proof}
Let $a$ denote the $d$-dimensional Gaussian random variable of standard deviation $\sigma$ and mean $0$.
Then $\|a\|_2$ has density $\chi_{d, \sigma}$.
By Lemma~\ref{lem:gaussdom}, $\|a\|_2$ is dominated by $\|b\|_2$. This implies that
$h(\|a\|_2)$ dominates $h(\|b\|_2)$ since $h$ is monotonically decreasing.
The lemma follows by observing that the two integrals are the two expected values of
$h(\|a\|_2)$ and $h(\|b\|_2)$.
\end{proof}

For Euclidean and squared Euclidean distances, it turns out to be useful to study
$\Delta_{a,b}(c) = d(c,a) - d(c,b)$ for points $a, b, c \in X$.
By abusing notation, we sometimes write $\Delta_{i,j}(k)$ instead of $\Delta_{X_i, X_j}(X_k)$
for short. A 2-change that replaces \edge 12\ and \edge 34\ by \edge 13\ and \edge 24\
improves the tour length by $\Delta_{1,4}(2) - \Delta_{1,4}(3)
 = \Delta_{2,3}(1) - \Delta_{2,3}(4)$.

\subsection{2-Opt State Graph and Linked 2-Changes}
\label{sec:linked}

The number of iterations that 2-opt needs depends of course heavily on the initial tour
and on which 2-change is chosen in each iteration.
We do not make any assumptions about the initial tour and about which
2-change is chosen. Following Englert et al.~\cite{EnglertEA:2Opt:2014},
we consider the \emph{2-opt state graph}: we have a node for every tour and a directed
edge from tour $T$ to tour $T'$ if $T'$ can be obtained
by one 2-change. The 2-opt state graph is a directed acyclic graph, and the length of the longest path in the 2-opt state
graph is an upper bound for the number of successful iterations that 2-opt needs.

In order to improve the bounds, we also consider \emph{pairs of
linked 2-changes}~\cite{EnglertEA:2Opt:2014}.
Two 2-changes form a pair of linked 2-changes if there is one edge added in one
2-change and removed in the other 2-change.
Formally, one 2-change replaces \edge 12\ and \edge 34\ by \edge 13\ and \edge 24\
and the other 2-change replaces \edge 13\ and \edge 56\ by \edge 15\ and \edge 26.
The edge \edge 13\ is the one that appears and disappears again (or the other way round).
It can happen
that \edge 24\ and \edge 56\ intersect.
Englert et al.~\cite{EnglertEA:2Opt:2014} called a pair of linked 2-changes a \emph{type $i$ pair}
if $|\edge 24 \cap \edge 56| = i$. As type 2 pairs, which involve only
four nodes, are difficult to analyze because of dependencies, we ignore them.
Fortunately, the following lemma states that we will find enough
disjoint pairs of linked 2-changes of type $0$ and $1$ in any sufficiently long sequence of 2-changes.

\begin{lemma}[\protect{Englert et al.~\cite[Lemma 9 of corrected version]{EnglertEA:2Opt:2014}}]
\label{lem:linkednumber}
Every sequence of $t$ consecutive 2-changes contains at least $t/7 - 3n/28$ disjoint
pairs of linked 2-changes of type 0 or type 1.
\end{lemma}

Following Englert et al.~\cite[Figure 8]{EnglertEA:2Opt:2014}, we subdivide type 1 pairs into type 1a and type 1b depending on
how \edge 24\ and \edge 56\ intersect. One of the 2-changes replaces \edge 12 and \edge 34\ by \edge 13\ and \edge 24.
Then other 2-change, i.e., the one that removes the edge \edge 13\ shared by the linked pair, determines its type:
\begin{description}
\item[Type 0:] $\edge 13$ and $\edge 56$ are replaced by $\edge 15$ and $\edge 36$.
\item[Type 1a:] $\edge 13$ and $\edge 25$ are replaced by $\edge 15$ and $\edge 23$.
\item[Type 1b:] $\edge 13$ and $\edge 25$ are replaced by $\edge 12$ and $\edge 35$.
\end{description}

The main idea in the proofs by Englert et al.~\cite{EnglertEA:2Opt:2014} and also in our proofs is
to bound the minimal improvement of any 2-change or
the minimal improvement of any pair of linked 2-changes. We denote the smallest improvement of any 2-change by $\dmin$
and the smallest improvement of any pair of linked 2-changes of type 0, 1a, or 1b by $\dminl$.
It will be clear from the context which distance measure is used
for $\dmin$ and $\dminl$.

Suppose that the initial tour has a length of at most $L$, then
2-opt cannot run for more than $L/\dmin$ iterations and
not for more than $\Theta(L/\dminl)$ iterations, provided that
$L/\dminl = \Omega(n^2)$ because of Lemma~\ref{lem:linkednumber}.

The following lemma formalizes this and shows how to bound the expected number of iterations using a tail bound for $\dmin$ or $\dminl$.

\begin{lemma}
\label{lem:generic}
Suppose that, with a probability of at least $1-1/n!$, any tour has a length of
at most $L$. Let $\gamma > 1$.
Then
\begin{enumerate}[label=(\arabic{*})]
\item If $\probab(\dmin \leq \eps) = O(P  \eps)$,
then the expected length of the longest path in the 2-opt state graph
is bounded from above by $O(PL n \log n)$. \label{gen:one}
\item If $\probab(\dmin \leq \eps) = O(P  \eps^\gamma)$,
then the expected length of the longest path in the 2-opt state graph
is bounded from above by $O(P^{1/\gamma}L)$. \label{gen:two}
\item The same bounds as (\ref{gen:one}) and (\ref{gen:two})
hold if we replace $\dmin$ by $\dminl$, provided that $PL = \Omega(n^2)$
for Case~\ref{gen:one} and $P^{1/\gamma}L = \Omega(n^2)$ for Case~\ref{gen:two}.
\label{gen:linked}
\end{enumerate}
\end{lemma}

\begin{proof}
If the length of the longest tour is longer than $L$, then we use the trivial
upper bound of $n!$. This contributes only $O(1)$ to the expected value.

Consider the first statement.
Let $T$ be the longest path in the 2-opt state graph. If $T \geq t$,
then $\dmin \leq L/t$.
Plugging this in and observing that $n!$ is an upper bound for $T$ yields
\[
  \expected(T)  = \sum_{t=1}^{n!} \probab(T\geq t) 
 \leq \sum_{t=1}^{n!} O(PL/t) = O(\log(n!) \cdot PL) = O(PLn \log n).
\]

Now consider the second statement, and let $T$ be as above.
Let $K = O(L \cdot P^{1/\gamma})$.
Then
\begin{align*}
  \expected(T) & = \sum_{t=1}^{n!} \probab(T \geq t) 
  \leq \sum_{t=1}^{n!} \min\left\{1, O\bigl(P \cdot (L/t)^\gamma\bigr)\right\} \\
& = K + PL^{\gamma} \cdot \sum_{t \geq K} O(t^{-\gamma})
 = K + PL^{\gamma} \cdot O(K^{1-\gamma}) = O(K).
\end{align*}

Finally, we consider the third statement.
The statement follows from the observation that
the maximal number of disjoint pairs of linked 2-changes and the length of the longest
path in the 2-opt state graph are asymptotically equal if they are of length at least $\Omega(n^2)$ (Lemma~\ref{lem:linkednumber}) and the probability statements become nontrivial
only for $t = \Omega(PL)$ in the first
and $t = \Omega(P^{1/\gamma}L)$ in the second case.
\end{proof}

\subsection{Manhattan Distances}
\label{sec:manhattan}

The essence of our analysis for Manhattan distances is a straightforward adaptation
of the analysis in the one-step model.
The extra factor of $\maxx$ comes from the bound of the initial tour, and the extra factor
of $d^2$ stems from stating the dependence on $d$ explicitly and getting rid of the exponential dependence on $d$~\cite[Proofs of Theorem~7 and Lemma~10]{EnglertEA:2Opt:2014}.

\begin{lemma}
\label{lem:manhattandminl}
$\probab(\dminl \leq \eps) = O(d^2n^6\eps^2/\sigma^2)$.
\end{lemma}

\begin{proof}
We consider a pair of linked 2-changes as described in Section~\ref{sec:linked}.
The improvement of the first 2-change is
\[
  \Gamma_1 = \sum_{i=1}^d | x_{1i} - x_{2i}| + | x_{3i} - x_{4i}| - | x_{1i} - x_{3i}| - | x_{2i} - x_{4i}|,
\]
where $x_{ji}$ is the $i$-th coordinate of $X_j \in X$.
The improvement of the second 2-change is
\[
  \Gamma_2 = \sum_{i=1}^d | x_{1i} - x_{3i}| + | x_{5i} - x_{6i}| - | x_{1i} - x_{5i}| - | x_{3i} - x_{6i}|.
\]
Note that we can have a type 1 pair, i.e., two of the points $X_2, X_4, X_5, X_6$ can be identical.

Each ordering of the $x_{ji}$ gives rise to a linear combination for $\Gamma_1$ and $\Gamma_2$. We have $(6!)^d$ such orderings.
If we examine the case distinctions by Englert et al.~\cite[Lemmas 11, 12, 13]{EnglertEA:2Opt:2014} closely,
we see that any pair of linear combinations is either impossible (it uses a different ordering of the variables for
$\Gamma_1$ and $\Gamma_2$ or one of $\Gamma_1$ and $\Gamma_2$ is non-positive, thus the corresponding
2-change is in fact not a 2-change)
or we have one variable $x_{ji}$ that has a non-zero coefficient in $\Gamma_1$ and a coefficient of $0$
in $\Gamma_2$ and another variable $x_{j'i'}$ that has 
a non-zero coefficient in $\Gamma_2$ and a coefficient of $0$
in $\Gamma_1$. The absolute values of the non-zero coefficients of $x_{ji}$
and $x_{j'i'}$ is $2$.
Now $\Gamma_1$ falls into $(0,\eps]$ only if $x_{ji}$ falls into an interval
of length $\eps/2$. This happens with a probability of at most $O(\eps/\sigma)$.
By independence, the same holds for $\Gamma_2$ and $x_{j'i'}$.

However, we would incur an extra factor of $(6!)^d$ in this way, and we would like to remove all exponential dependence of $d$.
In order to do this, we assume that we know $i$ and $i'$ already. This comes at the expense of a factor of $O(d^2)$ for taking
a union bound over the choices of $i$ and $i'$. We let an adversary fix values for all $x_{j \tilde i}$ with $\tilde i \neq i, i'$.
Since we know $i$ and $i'$, we are left with at most $(6!)^2 = O(1)$ possible linear combinations.

Finally, the lemma follows by taking a union bound over all $O(n^6)$ possible pairs of linked 2-changes.
\end{proof}

\begin{theorem}
\label{thm:manhattan}
The expected length of the longest path in the 2-opt state graph corresponding to $d$-dimensional instances
with Manhattan distances is at most $O(d^2 n^4 \maxx/\sigma)$.
\end{theorem}

\begin{proof}
The initial tour has a length of at most $O(nd\maxx)$ with a probability of at least $1-1/n!$
by Lemma~\ref{lem:maxx}.
We apply Lemma~\ref{lem:generic} for linked 2-changes using Lemma~\ref{lem:manhattandminl} and $\gamma = 2$.
\end{proof}

\subsection{Squared Euclidean Distances}
\label{sec:sed}

\subsubsection{Preparation}

In this section, we have $\Delta_{a,b}(c) = \|c-a\|^2_2 - \|c-b\|^2_2$
for $a, b, c \in \real^d$.

Assume that we have a 2-change that replaces $\edge 12$ and \edge 34\ by \edge 13\ and \edge 24.
The improvement caused by this 2-change is $\Delta_{2,3}(1) - \Delta_{2,3}(4) = \Delta_{1,4}(2) - \Delta_{1,4}(3)$.
Given the positions of the four nodes except for a single $i \in \{1,2,3,4\}$, such a 2-change yields a small improvement only if 
the corresponding $\Delta_{\cdot, \cdot}(i)$ falls into some interval of size $\eps$.
The following lemma gives an upper bound for the probability that this happens.

\begin{lemma}
\label{lem:seddelta}
Let $a, b \in \real^d$, $a \neq b$, and let $c$ be drawn according to a Gaussian
distribution with standard deviation $\sigma$. Let $I \subseteq \real$
be an interval of length $\eps$.
Then
\[
  \probab\bigl(\Delta_{a,b}(c) \in I\big) \leq \frac{\eps}{4 \sigma \cdot \|a-b\|_2}.
\]
\end{lemma}

\begin{proof}
Since Gaussian distributions are rotationally symmetric, we can assume without
loss of generality that $a = (0,\ldots, 0)$
and $b = (\delta, 0, \ldots, 0)$ with $\delta = \|a-b\|_2$.
Let $c = (c_1, \ldots, c_d)$. Then $\Delta_{a,b}(c) = c_1^2 - (c_1 - \delta)^2 = 2c_1\delta + \delta^2$.
Thus, $\Delta_{a,b}(c) \in I$ if and only
if $c_1$ falls into an interval of length $\frac\eps{2\delta}$.
Since $c_1$ is a 1-dimensional Gaussian random variable with a standard deviation
of $\sigma$, the probability for this is bounded from above
by $\frac{\eps}{4 \delta \sigma}$ since the maximum density of a 1-dimensional Gaussian of standard deviation $\sigma$ is
bounded from above by $\frac 1{2\sigma}$.
\end{proof}

\subsubsection{Single 2-Changes}
\label{ssec:sedsimple}

In this section, we prove a simple bound for the expected number of iterations of 2-opt with squared Euclidean distances.
This bounds holds for all $d \geq 2$. In the next section, we improve this bound for the case $d \geq 3$ using pairs of linked 2-changes.

\begin{lemma}
\label{lem:sedsingle}
For $d \geq 2$, we have $\probab(\dmin \in (0,\eps]) = O\bigl(\frac{n^4 \eps}{\sigma^2\sqrt d}\bigr)$.
\end{lemma}

\begin{proof}
Consider a 2-change where $\edge 12$ and $\edge 34$ are replaced by $\edge 13$ and $\edge 24$.
Its improvement is given by $\Delta_{2,3}(1) - \Delta_{2,3}(4)$.
We let an adversary fix $X_3$. Then we draw $X_2$. This fixes the distance $\delta = \|X_2 - X_3\|_2$.
Now we draw $X_4$. This fixes $\Delta_{2,3}(4)$. The 2-change yields an improvement of at most $\eps$
only if $\Delta_{2,3}(1)$ falls into an interval of size at most $\eps$.
According to Lemma~\ref{lem:seddelta},
the probability that this happens is at most $\frac{\eps}{4\delta \sigma}$.

Now let $g$ be the probability density of $\delta = \|X_2 - X_3\|$.
Then the probability that the 2-change yields an improvement of at most $\delta$ is bounded from above by
\[
\int_{0}^\infty g(\delta) \cdot \frac{\eps}{4 \sigma\delta} \, \textup d \delta 
\leq\int_{0}^\infty \chi_{d, \sigma}(\delta) \cdot \frac{\eps}{4 \sigma\delta} \, \textup d \delta 
= O\left(\frac{\eps}{\sigma^2 \sqrt d}\right).
\]
The first step is due to Lemma~\ref{lem:replacechi}. The second step is due to Lemma~\ref{lem:integral}
using $c=1$ and $d \geq 2$.
The lemma follows by a union bound over the $O(n^4)$ possible 2-changes.
\end{proof}

\begin{theorem}
\label{thm:squaredsingle}
For all $d \geq 2$, the expected length of the longest path in the 2-opt state graph corresponding to $d$-dimensional instances with
squared Euclidean distances is at most $O(\sqrt d \maxx^2 n^6 \log(n) /\sigma^2)$.
\end{theorem}

\begin{proof}
With a probability of at least $1-1/n!$, the instance is contained in a hypercube of sidelength $\maxx$. Thus, the longest
edge has a length of at most $\sqrt d \maxx$. Therefore, the initial tour has a length of at most $n d \maxx^2$.
We combine this with Lemmas~\ref{lem:generic} and~\ref{lem:sedsingle} to complete the proof.
\end{proof}

\subsubsection{Pairs of Linked 2-Changes}

We can obtain a better bound than in the previous section by analyzing pairs of linked 2-changes.
With the following three lemmas, we analyze the probability that pairs of linked 2-changes of type 0, 1a, or 1b
yield an improvement of at most $\eps$.

\begin{lemma}
\label{lem:pairs0sed}
For $d \geq 2$, the probability that there exists a pair of type 0 of linked 2-changes that yields an improvement of at most $\eps$ is bounded
from above by $O\bigl(\frac{n^6 \eps^2}{\sigma^{4}d}\bigr)$.
\end{lemma}

\begin{proof}
Consider a fixed pair of type 0 of linked 2-changes involving the six points $X_1, \ldots, X_6$ as described in Section~\ref{sec:linked}.
We show that the probability that it yields an improvement of at most $\eps$ is at most $O(\eps \sigma^{-2} / \sqrt d)$.
A union bound over the $O(n^6)$ possibilities of pairs of type 0 yields the lemma.

The basic idea is that we restrict ourselves to analyzing $\Delta_{1,4}(3)$ and $\Delta_{1,6}(5)$ only in order to
bound the probability that we have a small improvement.
In this way, we use the principle of deferred decision to show that we can analyze the improvements of the two 2-changes as if they were independent:
\begin{enumerate}[label=\arabic{*}.]
\item We let an adversary fix $X_1$ arbitrarily.
\item We draw $X_4$, which determines the distance $\|X_1 - X_4\|$.
\item We draw $X_2$. This fixes the position of the ``bad'' interval for $\Delta_{1,4}(3)$. Its size is already fixed since
we know the positions of $X_1$ and $X_4$. The position of $X_3$ is still random.
\item We draw $X_3$. The probability that
$X_3$ assumes a position such that the first 2-change yields an improvement of at most $\eps$ is thus at most $\frac{\eps}{4 \sigma \cdot \|X_1 - X_4\|}$.
\item We draw $X_6$. This determines the distance $\|X_1 - X_6\|$.
\item We draw $X_5$. The probability that $X_5$ assumes a position
such that the second 2-change yields an improvement of at most $\eps$ is thus at most $\frac{\eps}{4 \sigma \cdot \|X_1 - X_6\|}$.
\end{enumerate}

Let $g$ be the probability density function of the distance between $X_1$ and $X_4$, and let
$g'$ be the probability density function of the distance between $X_1$ and $X_6$.
By independence of the points, the probability that both 2-changes of the pair yield an improvement of at most $\eps$
is bounded from above by
\[
\int_{\delta=0}^\infty g(\delta) \cdot \frac{\eps}{4 \sigma \delta} \, \textup d \delta
\cdot 
\int_{\delta=0}^\infty g'(\delta) \cdot \frac{\eps}{4 \sigma \delta} \, \textup d \delta.
\]
We observe that $\frac{\eps}{4 \sigma \delta}$ is monotonically decreasing in $\delta$.
Thus, by Lemma~\ref{lem:replacechi}, we can replace $g$ and $g'$ by the density $\chi_{d,\sigma}$ of the chi distribution
to get the following upper bound for the probability that a pair of type 0 yields an improvement of at most $\eps$:
\[
\left(\int_{0}^{\infty} \chi_{d,\sigma}(\delta) \cdot \frac{\eps}{4 \delta \sigma} \, \textup d \delta \right)^2=
O\left(\frac{\eps^2}{\sigma^4 d}\right).
\]
Here, we use Lemma~\ref{lem:integral} with $c=1$, which is allowed since $d \geq 2$.
\end{proof}

\begin{lemma}
\label{lem:pairs1ased}
For $d \geq 2$, the probability that there exists a pair of type 1a of linked 2-changes that yields an improvement of at most $\eps$ is bounded
from above by $O\bigl(\frac{n^5 \eps^2}{\sigma^{4}d}\bigr)$.
\end{lemma}

\begin{proof}
We can analyze pairs of type 1a in the same way as type 0 pairs in Lemma~\ref{lem:pairs0sed}. To do this, we analyze
$\Delta_{2,3}(4)$ and $\Delta_{1,2}(5)$:
\begin{enumerate}[label=\arabic{*}.]
\item We let an adversary fix the position of $X_2$.
\item We draw $X_3$. This fixes $\|X_2 - X_3\|$.
\item We draw $X_1$. This fixes $\|X_1 - X_2\|$. In addition, this fixes the positions of the intervals
into which $\Delta_{2,3}(4)$ and $\Delta_{1,2}(5)$ must fall if the first or second 2-change yield an improvement of at most $\eps$.
\item We draw $X_4$.
\item We draw $X_5$.
\end{enumerate}
The remainder of the proof is identical to the proof of Lemma~\ref{lem:pairs0sed}, except that we have to take a union bound only
over $O(n^5)$ possible choices.
\end{proof}

\begin{lemma}
\label{lem:pairs1bsed}
For $d \geq 3$, the probability that there exists a pair of type 1b of linked 2-changes that yields an improvement of at most $\eps$ is bounded
from above by $O\bigl(\frac{n^5 \eps^2}{\sigma^{4}d}\bigr)$.
\end{lemma}

\begin{proof}
Again, we proceed similarly to Lemma~\ref{lem:pairs0sed}. We analyze a fixed pair of type 1b, where $\edge 12$ and $\edge 34$
are replaced by $\edge 13$ and $\edge 24$ in one step and $\edge 13$ and $\edge 25$ are replaced by $\edge 12$ and $\edge 35$,
and apply a union bound over the $O(n^5)$ possible type 1a pairs. We analyze the probability that $\Delta_{2,3}(4)$ or
$\Delta_{2,3}(5)$ assume a bad value.

We draw the points in the following order:
\begin{enumerate}[label=\arabic{*}.]
\item We fix $X_2$.
\item We draw $X_3$. This fixes the distance $\|X_2 - X_3\|_2$, which is crucial for both 2-changes.
\item We draw $X_1$.
\item We draw $X_4$. The probability that the first 2-change yields an improvement of at most $\eps$ is at most $\frac{\eps}{4 \sigma \cdot \|x_2 - x_3\|}$.
\item We draw $X_5$. The probability that the second 2-change yields an improvement of at most $\eps$ is at most $\frac{\eps}{4 \sigma \cdot \|x_2 - x_3\|}$.
\end{enumerate}

The main difference to Lemma~\ref{lem:pairs0sed} is that the sizes of the bad intervals are not independent.
However, once the size of the bad intervals is fixed, we can analyze the probabilities that $\Delta_{2,3}(4)$ or $\Delta_{2,3}(5)$
fall into their bad intervals as independent.
Given that $\|X_2 - X_3\| = \delta$ is fixed, the probability that the first and the second 2-change yield an improvement
of at most $\eps$ is bounded from above by $\frac{\eps^2}{16\delta^2\sigma^2}$.
Since this is decreasing in $\delta$, we can replace the distribution of $\delta$ by the chi distribution
to obtain an upper bound according to Lemma~\ref{lem:replacechi}.
Thus, using Lemma~\ref{lem:integral} with $c=2$ and $d \geq 3$, we obtain the following upper bound for the probability that
a pair of type 1b yields an improvement of at most $\eps$:
\[
\int_{\delta = 0}^\infty \chi_{d, \sigma}(\delta) \cdot \frac{\eps^2}{16\delta^2\sigma^2} \,\textup d \delta
 = O\left(\frac{\eps^2}{d\sigma^4}\right).
\]
\end{proof}

With the three lemmas above, we can obtain a bound on the expected number of iterations of 2-opt for TSP with squared Euclidean distances.

\begin{theorem}
\label{thm:sed}
For $d \geq 3$, the expected length of the longest path in the 2-opt state graph
corresponding to $d$-dimensional instances with squared Euclidean distances
is at most
$O\bigl(\frac{n^4 \sqrt d \maxx^2}{\sigma^2}\bigr)$.
\end{theorem}

\begin{proof}
The probability that any pair of linked 2-changes of type 0, 1a, or 1b yields an improvement of at most
$\eps$ is bounded from above by $O\bigl(\frac{\eps^2 n^6}{\sigma^4 d}\bigr)$.
We apply Lemma~\ref{lem:generic} with $\gamma = 2$ and observe that the initial tour has a length of at most $O(n d \maxx^2)$ with a
probability of at least $1 - 1/n!$.
\end{proof}

\subsection{Euclidean Distances}
\label{sec:euc}

\subsubsection{Differences of Euclidean Distances}

In this section, we have
$\Delta_{a,b}(z) = \|z-a\|_2 - \|z-b\|_2$
for $a, b, z \in \real^d$.
Analyzing $\|z-a\|_2 - \|z-b\|_2$ turns out to be more difficult than
analyzing $\|z-a\|^2_2 - \|z-b\|^2_2$ in the previous section. In particular the case
when $\|z-a\|_2 - \|z-b\|_2$ is close to its maximal value of $\|a-b\|_2$ requires special attention.
Intuitively, this is for the following reason: if $\Delta_{a,b}(z) \approx \|a-b\|_2$, then $z$ is close $L(a,b)$. Assume that $z \in L(a,b)$ for the moment.
Then either $z$ is between $a$ and $b$, which is fine. Or $z$ is not between $a$ and $b$. Then moving $z$ in the direction of $L(a,b)$ does not
change $\Delta_{a,b}(z)$ at all.

We observe that $\eta = \Delta_{a,b}(z)$ behaves essentially 2-dimensionally:
it depends only on the distance of $z$ from $L(a,b)$
(this is $x$ in the following lemma) and
on the position of the projection $z$ onto $L(a,b)$
(this is $y$ in the following lemma).
It also depends on the distance $\|a-b\|_2$ between $a$ and $b$
(this is $\delta$ in the following lemma, and we had this dependency also in the previous section about squared Euclidean distances).
The following lemma makes the connection between $x$ and $y$ explicit for a given $\eta$.
Figure~\ref{fig:implicit} depicts the situation described in the lemma.

\begin{lemma}
\label{lem:implicit}
Let $z = (x,y) \in \real^2$, $x\geq 0$, $y\geq 0$.
Let $a = (0, -\delta/2)$ and $b = (0, \delta/2)$ be two points at a distance of $\delta$.
Let $\eta = \|z-a\|_2 - \|z-b\|_2$. Then we have
\begin{equation}
  y^2 = \frac{\eta^2 \delta^2 + 4 \eta^2 x^2 - \eta^4}{4 \delta^2 - 4 \eta^2}
  = \frac{\eta^2}4 + \frac{\eta^2 x^2}{\delta^2 - \eta^2}
\label{expr:y}
\end{equation}
for $0 \leq \eta < \delta$
and
\begin{equation}
x^2 = \frac{y^2 \cdot \bigl(4 \delta^2 - 4 \eta^2) + \eta^4 - \eta^2 \delta^2}{4\eta^2}
=\frac{y^2 \cdot \bigl(\delta^2 - \eta^2)}{\eta^2} - \frac{\delta^2 - \eta^2}{4}.
\label{expr:x}
\end{equation}
for $\delta \geq \eta > 0$.
Furthermore, $\eta > \delta$ is impossible.
\end{lemma}

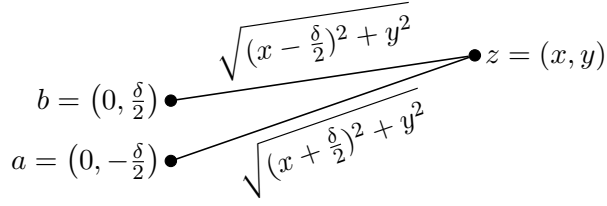
\begin{figure}[t]
\centering
\begin{tikzpicture}
\coordinate (pA) at (0, -0.4);
\coordinate (pB) at (0, 0.4);
\coordinate (pC) at (4, 1);
\node[Point] at (pA) (A) {};
\node[Point] at (pB) (B) {};
\node[Point] at (pC) (C) {}
    edge[Line] node[pos=0.5, sloped, above] {$\sqrt{(x-\frac \delta 2)^2 + y^2}$} (pB)
    edge[Line] node[pos=0.5, sloped, below] {$\sqrt{(x+\frac \delta 2)^2 + y^2}$} (pA);
\node[left] at (pA) {$a = \bigl(0, -\frac \delta 2\bigr)$};
\node[left] at (pB) {$b = \bigl(0, \frac \delta 2\bigr)$};
\node[right] at (pC) {$z=(x,y)$};
\end{tikzpicture}
\caption{The situation for Lemma~\ref{lem:implicit}.}
\label{fig:implicit}
\end{figure}

\begin{proof}
The last statement follows from the triangle inequality.

We have $\eta = \sqrt{(y+\delta/2)^2 + x^2} - \sqrt{(y-\delta/2)^2 + x^2}$.
Rearranging terms and squaring implies
\begin{align*}
 \eta^2 + (y-\delta/2)^2 + x^2 + 2 \eta \sqrt{(y-\delta/2)^2 + x^2} & = (y+\delta/2)^2 + x^2 \\
\Leftrightarrow 
 2 \eta \sqrt{(y-\delta/2)^2 + x^2} & = 2y\delta - \eta^2.
\end{align*}
Squaring again yields
\begin{align*}
 4 \eta^2 \cdot \bigl((y-\delta/2)^2 + x^2\bigr) & = 4y^2 \delta^2 - 4y\delta \eta^2 + \eta^4 \\
 \Leftrightarrow 4 \eta^2 \cdot \bigl(y^2 - \delta y + \delta^2/4 + x^2\bigr) & = 4y^2 \delta^2 - 4y\delta \eta^2 + \eta^4 \\
 \Leftrightarrow 4 \eta^2 y^2 - 4 \eta^2 \delta y + \eta^2 \delta^2 + 4 \eta^2 x^2 & = 4y^2 \delta^2 - 4y\delta \eta^2 + \eta^4 \\
 \Leftrightarrow 4 \eta^2 y^2 + \eta^2 \delta^2 + 4 \eta^2 x^2 & = 4y^2 \delta^2 + \eta^4.
\end{align*}
By rearranging terms again, we obtain
\[
y^2 \cdot \bigl(4 \delta^2 - 4 \eta^2) = \eta^2 \delta^2 + 4 \eta^2 x^2 - \eta^4.
\]
Using the assumption $\eta < \delta$ or $\eta>0$ implies the two claims.
\end{proof}

As said before, the difficult case in analyzing $\Delta_{a,b}(c) = \eta$ is when
$\eta \approx \|a-b\|_2$. In terms of the previous lemma, this can only happen if $x$ is small, i.e.,
if $c$ is close to $L(a,b)$, but not between $a$ and $b$.
The following lemmas makes a quantitative statement about this connection.

\begin{lemma}
\label{lem:newbadcone}
Let $a, b, z \in [-\maxx, \maxx]^d$. Assume that $\|a-b\| = \delta$ and that $z$ has a distance of $x$ from
$L(a,b)$. Then
\begin{equation}
  |\Delta_{a,b}(z)| \leq \delta - \frac{x^2\delta}{32d \maxx^2} . 
\label{newbadcone}
\end{equation}
\end{lemma}

\begin{proof}
Let $y$ be the distance of $z$ from $m = \frac{a+b}2$, and let $\eta = \Delta_{a,b}(z)$. Then, according to~\eqref{expr:x}, we have
\[
  x^2 = \frac{y^2 \cdot \bigl(\delta^2 - \eta^2)}{\eta^2} - \frac{\delta^2 - \eta^2}{4}.
\]
We have $\delta \leq \eta$. This and the upper bound $y \leq 2\sqrt d \maxx$ yields the following weaker bound:
\begin{equation}
  x^2 \leq \frac{4 d \maxx^2 \cdot \bigl(\delta^2 - \eta^2)}{\eta^2} . \label{simplified}
\end{equation}
We distinguish two cases.
The first case is that $\eta < \delta/2$.
In this case, it suffices to show that $\delta/2 \leq \delta - \frac{x^2 \delta }{32 d \maxx^2}$ in order to prove \eqref{newbadcone}.
Since $|x| \leq 2 \sqrt d \maxx$, this holds because
$\delta/2 \leq \delta - \delta/2$.

The second case is that $\eta \geq \delta/2$.
We have
\[
  \delta^2 - \eta^2 = (\delta + \eta) \cdot (\delta - \eta) \leq 2 \delta \cdot (\delta - \eta).
\]
Replacing $\delta^2-\eta^2$ by $2 \delta \cdot (\delta - \eta)$ in the numerator and $\eta^2$ by $\delta^2/4$ in the denominator of \eqref{simplified},
we obtain
\[
  x^2 \leq 
  \frac{4 d \maxx^2 \cdot \bigl(\delta^2 - \eta^2)}{\eta^2}  \leq \frac{32 d \maxx^2 \cdot (\eta - \delta)}{\delta}.
\]
Rearranging terms completes the proof.
\end{proof}

In order to be able to apply Lemma~\ref{lem:fctgauss}, we need the following
upper bound on the derivative of $y$ with respect to $\eta$, given that $x$ is fixed.

\begin{lemma}
\label{lem:derivative}
For $x, y \geq 0$, let
$y = \sqrt{\frac{\eta^2}4 + \frac{\eta^2x^2}{\delta^2 - \eta^2}}$ with $0 < \eta < \delta$.
Assume further that
$\eta \leq \delta - \frac{x^2\delta}{8d \maxx^2}$
and that $x \leq 2 \sqrt d \maxx$.
Then the derivative of $y$ with respect to $\eta$ is bounded by
\[
\frac 12 + \frac{32^{3/2}\maxx^3d^{3/2}}{\delta x^2} .
\]
\end{lemma}

\begin{proof}
The derivative of $y$ with respect to $\eta$ is given by
\begin{align*}
y' & = \frac{\textup d y}{\textup d \eta} = \frac{\delta^4 - 2 \delta^2\eta^2 + \eta^4 + 4 \delta^2 x^2}{2 \cdot (\delta^2 - \eta^2)^{3/2} \cdot \sqrt{\delta^2 - \eta^2 + 4 x^2}}\\
& = \frac{(\delta^2 - \eta^2)^2 + 4 \delta^2 x^2}{2 \cdot (\delta^2 - \eta^2)^{3/2} \cdot \sqrt{\delta^2 - \eta^2 + 4 x^2}} \\
& = \frac{(\delta^2 - \eta^2)^2}{2 \cdot (\delta^2 - \eta^2)^{3/2} \cdot \sqrt{\delta^2 - \eta^2 + 4 x^2}}
+\frac{4 \delta^2 x^2}{2 \cdot (\delta^2 - \eta^2)^{3/2} \cdot \sqrt{\delta^2 - \eta^2 + 4 x^2}} \\
& \leq \frac 12 
+\frac{2 \delta^2 x^2}{(\delta^2 - \eta^2)^{3/2} \cdot \sqrt{\delta^2 - \eta^2 + 4 x^2}}.
\end{align*}
We observe that $y' \geq 0$ for all $x$ and allowed choices of $\eta$ and $\delta$.
For the second term, we have
\[
\frac{2 \delta^2 x^2}{(\delta^2 - \eta^2)^{3/2} \cdot \sqrt{\delta^2 - \eta^2 + 4 x^2}}
\leq 
\frac{2 \delta^2 x^2}{(\delta^2 - \eta^2)^{3/2} \cdot \sqrt{4 x^2}}
= \frac{\delta^2 x}{(\delta^2 - \eta^2)^{3/2}}.
\]
By assumption, we have $\delta - \eta \geq \frac{x^2\delta}{32d \maxx^2}$
and $\eta \leq \delta$.
Thus, we have
\[
\frac{\delta^2 x}{(\delta^2 - \eta^2)^{3/2}} 
 = \frac{\delta^2 x}{\bigl((\delta + \eta)\cdot (\delta - \eta) \bigr)^{3/2} } 
 \leq \frac{\delta^2 x}{\delta^{3/2} \cdot \bigl(\frac{x^2\delta}{32d \maxx^2}\bigr)^{3/2}} 
  = \frac{32^{3/2}\maxx^3d^{3/2}}{\delta x^2}  .
\]
\end{proof}

Using Lemmas~\ref{lem:derivative} and~\ref{lem:fctgauss}, we can bound the probability
that $\Delta_{a,b}(z)$ assumes a value in an interval of size $\eps$.

%\todo{repair: points outside $[-\maxx, \maxx]^d$}

\begin{lemma}
\label{lem:differenceprob}
Let $d \geq 4$.
Let $a, b \in [-\maxx, \maxx]^d$ be arbitrary, $a \neq b$, and let
$z$ be drawn according to a Gaussian distribution with standard deviation $\sigma$.
Let $\delta = \|a-b\|_2$. Let $I$ be an interval of length $\eps$.
Then
\[
\probab\bigl(\Delta_{a,b}(z) \in I \bigr) \leq
  O\left(\frac{\eps \maxx^3 \sqrt d}{\sigma^3 \delta} \right) + \probab\bigl(z \notin [-\maxx, \maxx]^d\bigr).
\]
\end{lemma}

\begin{proof}
We assume throughout this proof that $z \in [-\maxx, \maxx]^d$. The case that this is not satisfied is taken care of by the
second term in the upper bound for the probability in the statement of the lemma.

Let $x$ denote the distance of $z$ to $L(a,b)$, and let $y$ denote the position
of the projection of $z$ onto $L(a,b)$.
First, let us assume that $x$ is fixed. Then, by Lemmas~\ref{lem:derivative}
and~\ref{lem:fctgauss}, the probability that
$\Delta_{a,b}(z) \in I$ is bounded from above by
\[
  O\left(\left(1 + \frac{\maxx^3d^{3/2}}{\delta x^2}\right) \cdot \frac{\eps}{\sigma}\right).
\]
Here, the requirements of Lemma~\ref{lem:derivative} are satisfied because of Lemma~\ref{lem:newbadcone}, or we have $z \notin [-\maxx, \maxx]^d$.

We observe that this probability is decreasing in $x$.
Thus, in order to get an upper bound for the probability with random $x$, we can use the $(d-1)$-dimensional chi distribution
for $x$ according to Lemma~\ref{lem:replacechi}.
We obtain
\begin{align*}
& \int_{x=0}^\infty \chi_{d-1, \sigma}(x) \cdot O\left(\left(1 + \frac{\maxx^3d^{3/2}}{\delta x^2}\right) \cdot \frac{\eps}{\sigma}\right) \, \text d x \\
= ~ & O\left(\frac \eps\sigma\right) + 
\int_{x=0}^\infty \chi_{d-1, \sigma}(x) \cdot O\left(\frac{\maxx^3d^{3/2}\eps}{\sigma\delta x^2}\right)  \, \textup d x 
= O\left(\frac \eps\sigma  +\frac{\maxx^3 \sqrt d \eps}{\sigma^3 \delta}    \right).
\end{align*}
by Lemma~\ref{lem:integral} using $c = 2$ and $d-1 \geq 3$.
Since $\delta \leq 2 \sqrt d \maxx$, the lemma follows.
\end{proof}

\subsubsection{Analysis of Pairs of 2-Changes}

We immediately go to pairs of linked 2-changes, as these yield the better bounds.

\begin{lemma}
\label{lem:euc0}
For $d \geq 4$, the probability that a pair of linked 2-changes of type 0 yields an improvement of at most
$\eps$ or some point lies outside $[-\maxx, \maxx]^d$ is bounded from above by
\[
O\left(\frac{n^6 \eps^2 \maxx^6}{\sigma^8}\right).
\]
\end{lemma}
\begin{proof}
We proceed similarly as in the proof of Lemma~\ref{lem:pairs0sed} for type 0 pairs for squared Euclidean distances.
We draw the points of a fixed pair of linked 2-changes as in the proof of Lemma~\ref{lem:pairs0sed}.

In the same way as in the proof of Lemma~\ref{lem:pairs0sed}, using Lemma~\ref{lem:differenceprob} instead of Lemma~\ref{lem:seddelta}, we obtain
that the probability that one fixed of the
two 2-changes yields an improved of at most $\eps$ is bounded from above by
\[
 \int_{\delta = 0}^\infty \chi_{d, \sigma}(\delta) \cdot 
 O\left(\frac{\eps \maxx^3 \sqrt d}{\sigma^3 \delta} \right) \, \textup d \delta
  = O\left(\frac{\eps \maxx^3}{\sigma^4} \right).
\]
Here, we applied Lemma~\ref{lem:integral} with $c = 1$.

Again in the same way as in the proof of Lemma~\ref{lem:pairs0sed}, we can analyze both 2-changes of the type 0 pair
is if they are independent. Finally, the lemma follows by a union bound over the $O(n^6)$ possibilities for a type 0 pair.
\end{proof}

\begin{lemma}
\label{lem:euc1a}
For $d \geq 4$, the probability that a pair of linked 2-changes of type 1a yields an improvement of at most
$\eps$ or some point lies outside $[-\maxx, \maxx]^d$ is bounded from above by
\[
O\left(\frac{n^5 \eps^2 \maxx^6}{\sigma^8}\right).
\]
\end{lemma}

\begin{proof}
The lemma can be proved in the same way as Lemma~\ref{lem:pairs1ased} with differences analogous to the proof of Lemma~\ref{lem:euc0}.
\end{proof}

\begin{lemma}
\label{lem:euc1b}
For $d \geq 4$, the probability that a pair of linked 2-changes of type 1b yields an improvement of at most
$\eps$ or some point lies outside $[-\maxx, \maxx]^d$ is bounded from above by
\[
O\left(\frac{n^5 \eps^2 \maxx^6}{\sigma^8}\right) .
\]
\end{lemma}

\begin{proof}
Similar to the proof of Lemma~\ref{lem:pairs1bsed} and using Lemma~\ref{lem:differenceprob}, the probability that
the two 2-changes of the pair both yield an improvement of at most $\eps$ is bounded from above by
\[
 \int_{\delta = 0}^\infty \chi_{d, \sigma}(\delta) \cdot O\left(\frac{\eps \maxx^3 \sqrt d}{\sigma^3 \delta} \right)^2 \, \textup d \delta
 = \int_{\delta = 0}^\infty \chi_{d, \sigma}(\delta) \cdot O\left(\frac{\eps^2 \maxx^6 d}{\sigma^6 \delta^2} \right) \, \textup d \delta.
\]
Now the lemma follows by applying Lemma~\ref{lem:integral} with $c = 2$.
\end{proof}

\begin{theorem}
For $d \geq 4$, the expected length of the longest path in the 2-opt state graph corresponding to $d$-dimensional instances
with Euclidean distances is at most $O(\frac{\sqrt d n^4 \maxx^4}{\sigma^4})$.
\end{theorem}

\begin{proof}
We have $\probab(\dminl \leq \eps) = O(\frac{n^6\eps^2 \maxx^6}{\sigma^8})$ by Lemmas~\ref{lem:euc0}, \ref{lem:euc1a}, and~\ref{lem:euc1b}.
If all points are in $[-\maxx, \maxx]^d$, then the longest edge has a length of $O(\sqrt d \maxx)$.
Thus, the initial tour has a length of at most $O(n \sqrt d \maxx)$.
Plugging this into Lemma~\ref{lem:generic} yields the result.
\end{proof}

%\chapter{Smoothed Analysis of the 2-Opt Heuristic}
%\label{chap:2opt}
%
%In this chapter, we analyze the approximation performance of the 2-Opt heuristic for Euclidean TSP (for an introduction, see \chapref{intropartapprox}). We start with introducing our notation, basic definitions and the formal description of our model of smoothed analysis in \secref{2optprelim}. We proceed to give general upper and lower bounds for the length of 2-optimal tours (i.e., locally optimal tours for the 2-Opt heuristic) in \secref{length}. Our upper bound of $O(\log(1/\sigma))$ is proved in \secref{upper}, which is complemented by an almost matching lower bound in \secref{lb}.
%

\section{Smoothed Analysis of the Approximation Ratio}
\label{sec:apx}

\subsection{Technical Preparation}

The following standard lemma provides a convenient way to bound the deviation of a perturbed point from its mean in the two-step model.

\begin{lemma}[\protect{Chi-square bound~\cite[Cor.\ 2.19]{SpielmanTeng:SmoothedAnalysisWhy:2004}}]\label{lem:chisquare}
Let $x$ be a Gaussian random vector in $\real^d$ of standard deviation $\sigma$ centered at the origin. Then, for $t\ge 3$, we have
%\[
 $\probab\bigl(\norm{x} \ge \sigma 3\sqrt{d\ln t}\bigr) \le t^{-2.9d}.$
%\]
\end{lemma}

To give large-deviation bounds on sums of independent variables with bounded support, we will make use of a standard Chernoff-Hoeffding bound.
\begin{lemma}[\protect{Chernoff-Hoeffding Bound~\cite[Exercise 1.1]{DubhashiPanconesi:ConcentrationOfMeasure:2009}}]\label{lem:chernoff}
Let $X := \sum_{i=1}^n X_i$, where $X_i, i=1,\dots,n$ are independently distributed in $[0,1]$, and $\mu_L \le \expected(X)\le \mu_H$. Then, for $0\le \varepsilon \le 1$, we have
\begin{align*}
 \probab\bigl(X > (1+\varepsilon) \cdot \mu_H \bigr) & \le \exp(-(\varepsilon^2/3) \cdot \mu_H),\\
 \probab\bigl(X < (1-\varepsilon) \cdot \mu_L \bigr) & \le \exp(-(\varepsilon^2/2) \cdot \mu_L).
\end{align*}
\end{lemma}

Throughout this paper, we assume that the dimension $d\ge 2$ is a fixed constant. Given a sequence of points $X_1,\dots,X_n \in \real^d$, we call a collection $T\subseteq [n]\times [n]$ of edges a \emph{tour}, if $T$ is connected and every $i\in[n] = \{1, \ldots, n\}$
has in- and outdegree exactly one in $T$. Note that we consider directed tours, which is useful in the analysis in this chapter,
but our distances are always symmetric.

Given any collection of edges $S$, its length is denoted by $L(S) = \sum_{(u,v)\in S} d(u,v)$, where $d(u,v)$ denotes
the Euclidean distance $\norm{X_u - X_v}$ between points $X_u$ and $X_v$.
 %We call a tour $T$ \emph{2-optimal}, if $d(u,v) + d(w,z) \le d(u,w) + d(v,z)$ for all edge
%pairs $(u,v),(w,z) \in T$. Equivalently, it is not possible to obtain a shorter tour by replacing $(u,v)$ and $(w,z)$ in a 2-optimal tour $T$ by two
%new edges.
%The 2-Opt heuristic replaces a pair of edges $(u,v)$ and $(w,z)$ by $(u,w)$ and $(v,z)$ if this decreases the tour length while this is possible.
%Thus, it terminates with a 2-optimal tour.

We call a collection $T \subseteq [n]^2$ a \emph{partial 2-optimal tour} if $T$ is a subset of a tour and $d(u,v) + d(w,z) \le d(u,w) + d(v,z)$ holds
for all edges $(u,v), (w,z) \in T$.
Our main interests are the traveling salesperson functional $\TSP(X) := \min_{\text{tour $T$}} L(T)$ as well as the functional $\TwoOPT(X):= \max_{\text{2-optimal tour $T$}} L(T)$ 
that maps the point set $X$ to the length of the longest 2-optimal tour through~$X$.

We note that the results in Section~\ref{sec:length} hold for metrics induced by arbitrary norms in $\real^d$ (Lemma~\ref{lem:2optworst} and~\ref{lem:2optSigma1}) or typical $\ell_p$ norms (Lemma~\ref{lem:phi-lbOPT} and~\ref{lem:sigma-lbOPT}), not only for the Euclidean metric.
We conjecture that also the upper bound in Section~\ref{sec:upper} holds for more general metrics, while the lower bound
in Section~\ref{sec:lb} is probably specific for the Euclidean metric. Still, we think that
the construction can be adapted to work for most natural metrics.

For obtaining lower bounds on the length of optimal tours, we consider the boundary functional $\bTSP(X)$ that attains the length of the shortest tour through all points in~$X$ that is allowed to traverse the boundary of $[0,1]^d$ at zero cost. For a proof of the following lemma, we refer to the monograph by Yukich~\cite{Yukich:ProbEuclidean:1998}.

\begin{lemma}[\protect{Boundary Functional~\cite[Lemma 3.7]{Yukich:ProbEuclidean:1998}}]\label{lem:boundary}
There is a constant $C > 0$ such that for all sets $X \subseteq [0,1]^d$ of $n$ points, we have $\bTSP(X) \ge \TSP(X) - C n^{\frac{d-2}{d-1}}$.
\end{lemma}

\subsection{Length of 2-Optimal Tours under Perturbations}
\label{sec:length}

In this section, we provide an upper bound for the length of any 2-optimal tour and a lower bound for the length of any global optimum. These
two results yield an upper bound of $O(1/\sigma)$ for the approximation ratio.

Chandra et al.~\cite{ChandraEA:OldOpt:1999}  proved a bound on the worst-case length of 2-optimal tours that, in fact, already holds for the more general notion of \emph{partial} 2-optimal tours. For an intuition why this is true, let us point out that their proof strategy is to argue that not too many long arcs in a tour may have similar directions due to the 2-optimality of the edges, while short edges do not contribute much to the length. The claim then follows from a packing argument. It is straight-forward to verify that it is never required that the collection of edges is closed or connected. %Note that it is important to regard the tour as a set of \emph{directed} edges.

\begin{lemma}[Length of partial 2-optimal tours \protect{\cite[Theorem 5.1]{ChandraEA:OldOpt:1999}}, paraphrased]
\label{lem:2optworst}
%Let $d\ge 2$.
%\todo{for which $d$?}
There exists a constant $c_d$ such that for every sequence $X$ of $n$ points in $[0,1]^d$, any partial 2-optimal tour has length less than $c_d \cdot n^{1-1/d}$.
\end{lemma}

While this bound directly applies to any perturbed instance under the one-step model, Gaussian perturbations fail to satisfy the premise of bounded
support in $[0,1]^d$. However, Gaussian tails are sufficiently light to enable us to translate the result to the two-step model by carefully taking care of outliers.

\begin{lemma}\label{lem:2optSigma1}
%Let $d\ge 2$.
%\todo{for which $d$?}
There exists a constant $b_d$ such that for any $\sigma \le 1$ the following statement holds.  For any $\bar{X}$, the probability that any partial 2-optimal tour on $X\leftarrow \pert_\sigma(\bar X)$ has length greater than $b_d \cdot n^{1-1/d}$, i.e.,
$\TwoOPT(X) > b_d\cdot n^{1-1/d}$, is bounded by $\exp(-\Omega(\sqrt{n}))$. Furthermore, 
\[
   \expected_{X\leftarrow \pert_\sigma(\bar X)}\bigl (\TwoOPT(X)\bigr) \le b_d\cdot  n^{1-1/d}.
\]
\end{lemma}
\begin{proof}
By translation, assume without loss of generality that the input points are contained in $[-\nicefrac{1}{2},\nicefrac{1}{2}]^d$. We define cubes $C_1=\ell_1[-1,1]^d,C_2=\ell_2[-1,1]^d,\ldots$ with $\ell_k := 6\sqrt{kd\ln 3}$. The side length of cube $C_k$ is $2 \ell_k$.
We consider the partitioning of $\real^d$ into the regions $C_1$ and $C_i \setminus C_{i-1}$ for $i \geq 2$. For some cube $C$ and any tour $T$, let $E_C(T)$ denote the edges in $T$ that are completely contained in $C$. For any tour $T$, the sequence $E_1, E_2, \dots$ defined by $E_1 := E_{C_1}(T)$ and $E_k := E_{C_k}(T) \setminus E_{C_{k-1}}(T)$, for $k\ge 2$, partitions the edges of $T$. Thus, $L(T) = \sum_{k=1}^\infty L(E_k)$.
%We consider the partitioning of $\real^d$ into These cubes partition the space $\real^d$ into the regions $C_1, C_2\setminus C_1,C_3\setminus C_2,\ldots$ For some cube $C$ and any tour $T$, let $E_C(T)$ denote the edges in $T$ that are completely contained in $C$. For any tour $T$, the sequence $E_1, E_2, \dots$ defined by $E_1 := E_{C_1}(T)$ and $E_k := E_{C_k}(T) \setminus E_{C_{k-1}}(T)$, for $k\ge 2$, partitions the edges of $T$. Thus, $L(T) = \sum_{k=1}^\infty L(E_k)$.

%For any outcome of the perturbations, let $T$ be a longest 2-optimal tour.
For any outcome of the perturbed points, let $T$ be the longest 2-optimal tour. Then, each $E_k$ is a partial 2-optimal tour in $C_k$. Let $n_k$ be the (random) number of points in $\real^d \setminus C_{k-1}$, which is an upper bound on the number of points in $C_k\setminus C_{k-1}$. At most $3n_k$ vertices are incident to the edges $E_k$, since each such edge is incident to at least one endpoint in $C_k \setminus C_{k-1}$ and every point has degree 2 in $T$. Since $C_k  = \ell_k [-1,1]^d$ is a translated unit cube scaled by $2\ell_k$, Lemma~\ref{lem:2optworst} yields
$L(E_k) \le c_d\cdot (2\ell_k) (3n_k)^{1-1/d}$.

%\[ L(\Tmax) = \sum_{k=1}^\infty L(E_k) \le \sum_{k=1}^\infty c_d 2\ell_k (2n_k)^{1-1/d}.\]
Observe that $X_i$ is not contained in $C_k$ only if its origin has been perturbed by noise of length at least $\ell_k/2$.
Thus, let $Z\sim \Gauss(0, \sigma^2)$ %be a $d$-dimensional Gaussian of mean $0$ and standard deviation $\sigma$.
and note that %$\ell_k \ge 1$ and 
$\sigma \le 1$ implies that $\ell_k/2 \ge 3\sqrt{dk\ln 3}\sigma$. Hence, for each point $X_i$, Lemma~\ref{lem:chisquare}
yields
\[
   \probab(X_i \notin C_k) \le \probab\left(\norm{Z} \ge \frac{\ell_k}{2}\right) \le 3^{-(2.9d)k}.
\]
%independently from all other $X_{i'}$.
By linearity of expectation, we conclude that $\expected(n_k) \le n3^{-(2.9d)(k-1)}$ for $k\ge 1$. %By Chernoff bounds, we conclude that $n_k\le 2n3^{-(2.9d)(k-1)^2}$ with probability at least $1-\exp()$.
This yields
\begin{align*} 
\expected\bigl(L(T)\bigr) & =  \sum_{k=1}^\infty \expected(L(E_k)) 
 \le  \sum_{k=1}^\infty c_d \cdot (2\ell_k) (3\expected(n_k))^{1-1/d}\\
& \le  c_d \cdot 12\sqrt{d\ln 3}\cdot (3n)^{1-1/d} \left(\sum_{k=1}^\infty \sqrt{k} 3^{-2.9(d-1)(k-1)}\right) = O(n^{1-1/d}),
\end{align*} 
where we used Jensen's inequality for the first inequality.

To derive tail bounds for the length of any 2-optimal tour, let $N_k := n3^{-2.9d(k-1)}$ be the upper bound on $\expected(n_k)$ derived above. By the Chernoff bound (Lemma~\ref{lem:chernoff}), we have
\[ \probab(n_k \ge 2N_k) \le \exp(-N_k/3).\]
This guarantee is only strong as long as $N_k$ is sufficiently large. Hence, we regard this guarantee only for $1\le k \le k_1$, where $k_1$ is chosen such that $\sqrt{n} \le N_{k_1}\le 3^{2.9d}\sqrt{n}$. Assume that $n_k \le 2N_k$ for all $1\le k \le k_1$. Then, analogously to the above calculation, the contribution of $E_1, \dots, E_{k_1}$ is bounded by
\begin{align*}
  \sum_{k=1}^{k_1} L(E_k) & \le \sum_{k=1}^{k_1} c_d \cdot (2\ell_k) (6N_k)^{1-1/d} \\
& \le c_d \cdot 12\sqrt{d\ln 3}\cdot  (6n)^{1-1/d} \left(\sum_{k=1}^\infty \sqrt{k} 3^{-2.9(d-1)(k-1)}\right) = O(n^{1-1/d}).
\end{align*}
Let $p_1$ denote the probability that some $1\le k \le k_1$ fails to satisfy $n_k \le 2N_k$. Then,
\[ p_1 \le \sum_{k=1}^{k_1} \probab(n_k > 2N_k) \le k_1 \exp(-N_{k_1}/3) = \exp(-\Omega(\sqrt{n})).\]
Let us continue assuming that all $1\le k \le k_1$ satisfy $n_k \le 2N_k$. Since in particular $n_{k_1} \le 2 N_{k_1}$, at most $n_{k_1} \le 2\cdot 3^{2.9d}\sqrt{n}$ vertices remain outside $C_{k_1-1}$. Let $k_2 := \lceil \sqrt{n}\rceil $. By a union bound,
\[ p_2:= \probab(\exists j: X_j \notin C_{k_2}) \le n3^{-2.9d(k_2-1)} = \exp(-\Omega(\sqrt{n})). \]
Assume that the corresponding event holds (i.e., $X\subseteq C_{k_2}$), then the remaining points outside $C_{k_1-1}$ (and hence, outside $C_{k_1}$) are contained in $C_{k_2}$. We conclude that, with probability at least $1-(p_1+p_2)=1-\exp(-\Omega(\sqrt{n}))$, we have
\begin{align*}
 \sum_{k=k_1+1}^\infty L(E_k) & = \sum_{k=k_1+1}^{k_2} L(E_k) \le c_d \cdot (2\ell_{k_2}) (6N_{k_1})^{1-1/d} \\
& = O(\sqrt{k_2}(N_{k_1})^{1-1/d}) = O(n^{\frac{1}{4}+\frac{1}{2}\left(1-\frac{1}{d}\right)}) = O(n^{1-1/d}).
\end{align*}
This finishes the claim, since we have shown that with probability $1-\exp(-\Omega(\sqrt{n}))$, both the contribution of $E_1,\dots,E_{k_1}$ and $E_{k_1+1},\dots$ is bounded by $O(n^{1-1/d})$.
\end{proof}

We complement the bound above by a lower bound on tour lengths of perturbed inputs, making use of the following result by Englert et al.~\cite{EnglertEA:2Opt:2014} for the one-step model.
\begin{lemma}[\protect{Englert et al.~\cite[Proof of Theorem 1.4]{EnglertEA:2Opt:2014}}]\label{lem:phi-lbOPT}
Let $X_1,\dots,X_n$ be a $\phi$-perturbed instance. Then with probability $1-\exp(-\Omega(n))$, any tour on $X_1,\dots,X_n$ has length at least $\Omega(n^{1-1/d}/\sqrt[d]{\phi})$.
\end{lemma}

It also follows from their results that this bound translates to the two-step model consistently with the intuitive correspondence of $\phi \sim \sigma^{-d}$ between the one-step and the two-step model. 
\begin{lemma}\label{lem:sigma-lbOPT}
Let $X_1,\dots,X_n$ be an instance of points in the unit cube perturbed by Gaussians of standard deviation $\sigma \le 1$. Then with probability $1-\exp(-\Omega(n))$ any tour on $X_1,\dots,X_n$ has length at least $\Omega(\sigma n^{1-1/d})$.
\end{lemma}
\begin{proof}
We summarize the arguments of Englert et al.~\cite[Section 6]{EnglertEA:2Opt:2014} first, who considered truncated Gaussian perturbations: Here, we condition the Gaussian perturbation $Z_i$ for each input point $X_i$ to be contained in $A := [-\alpha, \alpha]^d$ for some $\alpha \ge 1$. Conditioned on this event, the resulting input instance is contained in the cube $C := [-\alpha,1+\alpha]^d$. By straight-forward calculations, the conditional distribution of each point in $C$ has maximum density bounded by $O(\alpha^d/\sigma^d)$. Moreover, the probability that the condition fails for a single point is bounded by $\probab(Z_i \notin A) \le d\sigma \exp(-\alpha^2/(2\sigma^2))$ for all $i$. Thus, by choosing $\alpha\ge 1$ sufficiently large, each point has at least constant probability to satisfy the condition $Z_i\in A$.

Given any instance (with Gaussian perturbations which are not truncated), first reveal the (random) subinstance of those points for which the condition $Z_i \in A$ is satisfied and let $n'$ be the number of such points. By the Chernoff bound (Lemma~\ref{lem:chernoff}), and $\probab(Z_i\in A)= \Omega(1)$, we have $n' \ge c\cdot n$ for some $c>0$ with probability at least $1-\exp(-\Omega(n))$. If this event occurs, we obtain a random instance of $n'\ge cn$ points and maximum density $\phi = O(\alpha^d/\sigma^d)$. Hence an application of Lemma~\ref{lem:phi-lbOPT} yields that, for some constant $c'>0$, the probability that a tour of length less than $c' \cdot (n')^{1-1/d}/\sqrt[d]{\phi} = O((\sigma/\alpha) n^{1-1/d})=O(\sigma n^{1-1/d})$ exists is at most $\exp(-\Omega(n))+\exp(-\Omega(n'))=\exp(-\Omega(n))$.
\end{proof}

Note that Lemmas~\ref{lem:2optSigma1} and~\ref{lem:sigma-lbOPT} almost immediately yield the following bound on the approximation performance for the two-step model.
(The large-deviation bound is immediate. For the expected approximation ratio, we make use of the worst-case bound of $O(\log n)$, given in \lemref{logworstcase}
below.)
%\footnote{The large-deviation bound is immediate. Showing the expected approximation ratio is also not difficult (we would additionally need Lemma~\ref{lem:logworstcase}, proven in \secref{upper}), but is omitted here, since the aim is to obtain a stronger result in \secref{upper} anyway.}
\begin{observation}\label{obs:easyApprox}
Let $X_1,\dots,X_n$ be an instance of points in the unit cube perturbed by Gaussians of standard deviation $\sigma \le 1$. Then the approximation performance of 2-Opt is bounded by $O(1/\sigma)$ in expectation and with probability $1-\exp(-\Omega(\sqrt{n}))$.
\end{observation}

We remark that this bound is best possible for an analysis of perturbed instances that separately bounds the lengths of any 2-optimal tour from above and gives a lower bound on any optimal tour. To see this, we argue that Lemma~\ref{lem:phi-lbOPT}, Lemma~\ref{lem:2optworst} (even under $\phi$-perturbed input),
Lemma~\ref{lem:sigma-lbOPT} and Lemma~\ref{lem:2optSigma1} cannot be improved in general. This is straight-forward for Lemma~\ref{lem:phi-lbOPT}, since $n$ points distributed uniformly at random in a cube of volume $1/\phi$ always have, by scaling and Lemma~\ref{lem:2optworst}, a tour of length $O(n^{1-1/d}/\sqrt[d]{\phi})$. Hence, the lower bound on optimal tours on perturbed instances is tight. To see that the upper bound on any 2-optimal tour is tight, take $n$ uniformly distributed points that have, by Lemma~\ref{lem:phi-lbOPT}, an optimal tour of length $\Omega(n^{1-1/d})$ with high probability and thus also in expectation.

Naturally, this transfers to the case of Gaussian perturbations, albeit more technical to verify: If we place $n$ identical points in $[0,1]^d$, say at the origin, and perturb them with Gaussians of standard deviation $\sigma$, then we may without loss of generality scale the unit cube to $[0,1/\sigma]^d$ and perturb the points with standard deviation $1$ instead. By Lemma~\ref{lem:2optSigma1}, any 2-optimal tour and, thus, any optimal tour on these points has a length of $O(n^{1 - 1/d})$ on the scaled instance, since the origins are still contained in the unit cube.
Thus, the optimal tour on the original instance has a length of at most $O(\sigma \cdot n^{1-1/d})$ in expectation and with high probability.

We only sketch that 2-optimal tours can have a length of at least $\Omega(n^{1- 1/d})$:
We distribute the $n$ (unperturbed) points into $1/\sigma^d$ groups of $\sigma^d n$ points each, and we partition the cube $[0,1]^d$ into $1/\sigma^d$
subcubes of equal side length. Let $c > 0$ be a constant such that with high probability, at least $c \sigma^d n$ points of a group remain in their subcube after perturbation.
We call these points successful.
Since successful points are identically distributed, conditioned on falling into a compact set, the shortest tour through these (at least) $c \sigma^d n$ points has a length of at least $\sigma\cdot c'(\sigma^d n)^{1-1/d}=c'\sigma^d n^{1-1/d}$ for some other constant $c'> 0$~\cite{Yukich:ProbEuclidean:1998}.
(This is just a scaled version of perturbing and truncating a Gaussian of standard deviation $1$ to a unit hypercube, which would result in a tour length
of $m^{1-1/d}$ for $m$ points.)
By closeness of the tour on all points to the boundary functional and geometric superadditivity of the boundary functional
(see Yukich~\cite{Yukich:ProbEuclidean:1998} for details), it follows that the optimal tour on all successful points is at least
$\Omega\bigl((1/\sigma^d) \cdot \sigma^dn^{1-1/d}\bigr) = \Omega(n^{1-1/d})$.
 
\subsection{Upper Bound on the Approximation Performance}
\label{sec:upper}

In this section, we establish an upper bound on the approximation performance of 2-Opt under Gaussian perturbations.
We achieve a bound of $O(\log 1/\sigma)$. 
Due to the lower bound presented in Section~\ref{sec:lb}, improving the smoothed approximation ratio to $o\bigl(\log(1/\sigma)/\log\log(1/\sigma)\bigr)$
is impossible.
Thus, our bound is almost tight.

As noted in the previous section, to beat $O(1/\sigma)$ it is essential to exploit the structure of the unperturbed input. This will be achieved by classifying edges of a tour into \emph{long} and \emph{short} edges and bounding the length of long edges by a (worst-case) global argument and short edges locally against the partial optimal tour on subinstances (by a reduction to an (almost-)average case). The local arguments for short edges will exploit how many unperturbed origins lie in the vicinity of a given region.

The global argument bounding long edges follows from the worst-case $O(\log n)$ bound on the worst-case approximation performance~\cite{ChandraEA:OldOpt:1999} that we rephrase here for our purposes.

\begin{lemma}[\protect{\cite[Proof of Theorem 4.3]{ChandraEA:OldOpt:1999}}]\label{lem:logworstcase}
Let $T$ be a 2-optimal tour and $\OPT$ denote the length of the optimal traveling salesperson tour $T_\OPT$. Let $T_i$ contain the set of all edges in $T$ whose length is in $[\OPT/2^i,\OPT/2^{i-1}]$. Then $L(T_i) = O(\OPT)$. In particular, it follows that $L(T)= O(\log n)\cdot \OPT$.
\end{lemma}

In the proof of our bound of $O(\log 1/\sigma)$, the above lemma accounts for all edges of length $[\Omega(\sigma),O(1)]$.
A central idea to bound all shorter edges is to apply the one-step model result to small parts of the input space. In particular, we will condition sets of points to be perturbed into cubes of side length $\sigma$. The following technical lemma helps to capture what values of $\phi$ suffice to express the conditional density function of these points depending on the distance of their unperturbed origins to the cube. This allows for appealing to the one-step model result of Lemma~\ref{lem:phi-lbOPT}.

\begin{lemma}\label{lem:squaredens}
Let $c\in [0,\sigma]^d$ and $k=(k_1,\dots,k_d) \in \mathbb{N}_0^d$. Let $Y$ be the random variable $X\sim \Gauss(c,\sigma^2)$ conditioned on $X \in Q:=[k_1\sigma,(k_1+1)\sigma]\times \cdots \times[k_d\sigma,(k_d+1)\sigma]$ and $f_Y$ be the corresponding probability density function. Then $f_Y$ is bounded from above by $\exp(\norm{k}_1+(3/2)d)\sigma^{-d}$.
\end{lemma}
\begin{proof}
Let $f_X(x)=\frac{1}{(2\pi)^{d/2}\sigma^d}\cdot \exp(-\frac{\norm{x-c}^2}{2\sigma^2})$ be the probability density function of $X$. Let $q:=\argmin_{z\in Q} \norm{z-c}$ be the point in $Q$ that is closest to $c$. Then, since $f_X(x)$ is rotationally invariant around $c$ and decreasing in $\norm{x-c}$, the density $f_X(x)$ inside~$Q$ is maximized at $x=q$. Likewise, $q':=\argmax_{z\in Q}\norm{z-c}$ minimizes the density inside $Q$. Since $Q$ is a $(\sigma\times \cdots \times \sigma)$-cube in $\mathbb{R}_{\ge 0}$, $\norm{q'-c} \le \norm{(q + \sigma \one)-c}$, where $\one= (1,\dots,1) \in \real^d$ denotes the all-ones vector. Given $g(q):=\frac{f_X(q)}{f_X(q+\sigma \one)}$, we can thus bound the conditional probability density function $f_Y$ for $x\in Q$ by
\[f_Y(x) = \frac{f_X(x)}{\int_Q f_X(y) dy} \le \frac{f_X(x)}{f_X(q')\vol(Q)} \le \frac{f_X(q)}{f_X(q+\sigma \one)}\cdot \frac{1}{\sigma^d} =\frac{g(q)}{\sigma^d}.\]
It remains to bound, for $x\in Q$, 
\[g(x) = \frac{f(x)}{f(x+\sigma \one)} =  \exp\left(\frac{\norm{(x-c)+\sigma \one}^2-\norm{x-c}^2}{2\sigma^2}\right) = \exp\left(\frac{\norm{x-c}_1}{\sigma}+\frac{d}{2}\right).\]
Since for all $x\in Q$, $\norm{x-c}_1 \le (\norm{k}_1 + d)\sigma$, we can bound $g(x)\le \exp(\norm{k}_1 + (3/2)d)$, yielding the claim.
\end{proof}

The main result of this section is the following theorem, which will be proved in the remainder of the section.

\begin{theorem}\label{thm:upper}
Let $X=(X_1,\dots,X_n)$ be an instance of points in $[0,1]^d$ perturbed by Gaussians of standard deviation $\sigma \le 1$.  With probability $1-\exp(-\Omega(n^{1/2-\varepsilon}))$ for any constant $\varepsilon > 0$, we have $\TwoOPT(X) = O(\log(1/\sigma)) \cdot \TSP(X)$.
Furthermore,
\[\expected\left(\frac{\TwoOPT(X)}{\TSP(X)}\right) = O(\log(1/\sigma)).
\]
\end{theorem}

Since the approximation performance of 2-Opt is bounded by $O(\log n)$ in the worst-case, we may assume that $1/\sigma = O(n^\varepsilon)$ for all constant $\varepsilon>0$, since otherwise our smoothed result is superseded by Lemma~\ref{lem:logworstcase}. Furthermore, we may also assume that $1/\sigma = \omega(1)$, since otherwise \obsref{easyApprox} already yields the result. In what follows, let $T_\OPT$ and $T$ be any optimal and longest 2-optimal, respectively, traveling salesperson tour on $X_1,\dots,X_n$. Furthermore, we let $\OPT = L(T_\OPT)$ denote the length of the shortest traveling salesperson tour.

\subsubsection{Outliers and Long Edges}

We will first show that the contribution of almost all points outside $[0,1]^d$ is bounded by $O(\sigma n^{1-1/d})$ with high probability and in expectation, similar to Lemma~\ref{lem:2optSigma1}. For this, we define growing cubes $A_i := [-a_i, 1+a_i]^d$, where we set $a_i := 3\sigma \sqrt{di\ln(3/\sigma)}$ for $i\ge 1$ and $A_0 = [0,1]^d$. Let $n_i$ be the number of points not contained in $A_{i-1}$. For every point $X_j$, Lemma~\ref{lem:chisquare} with $t:=(3/\sigma)^i$ bounds $\probab(X_j \notin A_{i-1}) \le (\sigma/3)^{2.9d(i-1)}$ (note that we have chosen the $a_i$ such that $t\ge 3$).
Thus, $\expected(n_i)\le n(\sigma/3)^{2.9d(i-1)}$. 
We define $E_i$ as the set of edges of the longest 2-optimal tour $T$ contained in $A_i$ with at least one endpoint in $A_i\setminus A_{i-1}$. We first bound the contribution of the $E_i$ with $i\ge 2$. 

\begin{lemma}\label{lem:outliers}
With probability $1-\exp(-\Omega(n^{1/2-\varepsilon}))$ for any constant $\varepsilon > 0$, 
we have
\[
 \sum_{i=2}^\infty L(E_i)= O\left(\sigma n^{1-1/d}\right).
\]
In addition, we have $\expected\left(\sum_{i=2}^\infty L(E_i)\right)= O\left(\sigma n^{1-1/d}\right)$.
\end{lemma}
\begin{proof}
The proof is analogous to the proof of \lemref{2optSigma1}. Linearity of expectation, \lemref{2optworst}, and Jensen's inequality yield
\begin{align*}
\sum_{i=2}^\infty \expected(L(E_i)) & \le  \sum_{i=2}^\infty c_d\cdot (3\expected(n_i))^{1-1/d}(1+2a_i) \\
& \le  \sum_{i=2}^\infty 3c_d \cdot n^{1-1/d}\left(\frac{\sigma}{3}\right)^{2.9(d-1)(i-1)} \left(1+6\sigma \sqrt{i\ln(3/\sigma)}\right)\\
& \le  3c_d\cdot n^{1-1/d}\left(\frac{\sigma}{3}\right)^{2.9(d-1)} (1+6\sigma\sqrt{\ln(3/\sigma)}) \left(
   \sum_{i=0}^\infty \sqrt{i+2}\left(\frac{\sigma}{3}\right)^{2.9(d-1)i}\right).
\end{align*}
By observing that $\sum_{i=0}^\infty \sqrt{i+2}(\sigma/3)^{2.9(d-1)i}$ is bounded by a constant, we conclude that $\sum_{i=2}^\infty \expected(L(E_i))$ is bounded by $O(\sigma n^{1-1/d})$. 

Let $N_k := n(\sigma/3)^{2.9d(k-1)}$ be the upper bound on $\expected(n_k)$ derived above. By the Chernoff bounds (Lemma~\ref{lem:chernoff}), we have
\[
  \probab(n_k \ge 2N_k) \le \exp(-N_k/3).
\]
Choose $k_1$ such that $(\sigma/3)^{2.9d} \sigma \sqrt{n} \le N_{k_1}\le \sigma\sqrt{n}$. Thus, $k_1 = O(\log n)$.  Assume that $n_k \le 2N_k$ for all $1\le k \le k_1$. Then, analogously to the above calculation, the contribution of $E_2, \dots, E_{k_1}$ is bounded by
\begin{align*}
 \sum_{k=2}^{k_1} L(E_k) & \le  \sum_{k=2}^{k_1} c_d \cdot (1+2a_k) (6N_k)^{1-1/d} \\
& \le  c_d \cdot \left(1+6\sigma \sqrt{d\ln( 3/\sigma)}\right)\cdot  (6n)^{1-1/d} (\sigma/3)^{2.9(d-1)} \\
& \qquad \cdot \left(\sum_{k=0}^\infty \sqrt{k+2} (\sigma / 3)^{2.9(d-1)k}\right)  =  O(\sigma n^{1-1/d}).
\end{align*}
Note that the probability that some $1\le k \le k_1$ fails to satisfy $n_k \le 2N_k$ is bounded by 
\[ \sum_{k=1}^{k_1} \probab(n_k > 2N_k) \le k_1 \exp(-N_{k_1}/3) = \exp(-\Omega(n^{1/2-\varepsilon})),\]
for any constant $\varepsilon > 0$. 
Since $n_{k_1} \le 2 N_{k_1}$, at most $n_{k_1} \le 2\sigma\sqrt{n}$ vertices remain outside $A_{k_1-1}$. Let $k_2 := \lceil \sigma \sqrt{n} \rceil$. By a
union bound, for any constant $\varepsilon >0$,
\[ \probab(\exists j: X_j \notin A_{k_2}) \le n(\sigma/3)^{2.9d(k_2-1)} = \exp(-\Omega(n^{1/2-\varepsilon})). \]
Assume that we have the -- very likely -- event that all points are in $A_{k_2}$, then the remaining points outside $A_{k_1-1}$ are contained in $A_{k_2}$. We conclude that
\begin{align*}
\sum_{k=k_1}^\infty L(E_k) & = \sum_{k=k_1}^{k_2} L(E_k) \\
& \le c_d \cdot (1+2a_{k_2}) (6N_{k_1})^{1-1/d} \\
& = O(\sqrt{k_2}(N_{k_1})^{1-1/d}) \\
& = O(\sigma^{3/2-1/d} n^{\frac{1}{4}+\frac{1}{2}\left(1-\frac{1}{d}\right)}) = O(\sigma n^{1-1/d}). \qedhere
\end{align*}
\end{proof}

In the remainder of the proof, we bound the total length of edges inside~$A_1$. Define $C:=A_1$ and note that all edges in $C$ have bounded length $\sqrt{d}(1+2a_1)=O(1)$. We let $T_i$ contain the set of all those edges within $C$ (in the longest 2-optimal tour $T$) whose lengths are in $[\OPT/2^i,\OPT/2^{i-1}]$. Let $k_1$ be such that $\sqrt{d}(1+2a_1) \in [\OPT/2^{k_1},\OPT/2^{k_1-1}]$. Then $L(T_k)=0$ for all $k <k_1$, since no longer edges exist. Let $k_2$ be such that $\sigma \in [\OPT/2^{k_2},\OPT/2^{k_2-1}]$. Then $\sum_{k=k_1}^{k_2} L(T_k) = O((k_2-k_1)\cdot \OPT) = O(\log(1/\sigma)\OPT)$ by Lemma~\ref{lem:logworstcase}. This argument bounds the contribution of \emph{long edges}, i.e., edges longer than $\sigma$, in the worst case, after observing the perturbation of the input points. It remains to bound the length of short edges in~$C$, which we do in the next section.

\subsubsection{Short Edges}

To account for the length of the remaining edges, we take a different route than for the long edges: Call an edge that is shorter than $\sigma$ a \emph{short edge} and partition the bounding box $C=[-a_1,1+a_1]^d$ into a grid of $(\sigma \times \cdots \times \sigma)$-cubes $C_1,\dots, C_M$ with $M = \Theta((\sigma/(1+a_1))^{-d})=\Theta(\sigma^{-d})$, which we call \emph{cells}. All edges in $T_k$ for $k\ge k_2$, i.e., short edges, are completely contained in a single cell or run from some cell $C_i$ to one of its $3^d-1$ neighboring cells.  For a given tour $T$, let $E_{C_i}(T)$ denote the short edges of $T$ for which at least one of the endpoints lies in $C_i$. 

We aim to relate the length of the edges $E_{C_i}(T)$ for the longest 2-optimal tour $T$ to the length of the edges $E_{C_i}(T_\OPT)$ of the optimal tour $T_\OPT$. This local approach is justified by the following property.

\begin{lemma}\label{lem:boundTSP}
For any tour $T'$, the contribution $L(E_{C_i}(T'))$ of cell $C_i$ is lower bounded by $\TSP(X\cap C_i) - O(\sigma {|X\cap C_i|}^{\frac{d-2}{d-1}})$.
\end{lemma}
\begin{proof}%[of Lemma~\ref{lem:boundTSP}]
Consider all edges $S$ in $T'$ that have at least one endpoint in $C_i$. Replacing those edges $(u,v) \in S$ with $u\in C_i$ and $v\notin C_i$ by the shortest edge connecting $u$ to the boundary of $C_i$ does not increase the total edge length by triangle inequality. If $C_i$ were the unit cube, $L(E_{C_i}(T'))$ would thus be lower bounded by the boundary functional $\TSP^B(X\cap C_i)$. Instead, we scale the instance $X\cap C_i$ by $1/\sigma$ to obtain an instance $X'$ in the unit cube, satisfying $\TSP(X\cap C_i) = \sigma \TSP(X')$ and, as argued above, $L(E_{C_i}(T')) \ge \sigma \TSP^B(X')$. Thus an application of Lemma~\ref{lem:boundary} yields 
\[ L(E_{C_i}(T')) \ge \sigma \left(\TSP(X') - O\bigl({|X'|}^{\frac{d-2}{d-1}}\bigr)\right) = \TSP(X\cap C_i) - O\bigl(\sigma \cdot {|X\cap C_i|}^{\frac{d-2}{d-1}}\bigr).\qedhere\]
\end{proof}

Intuitively, a cell $C_i$ is of one of two kinds: either few points are expected to be perturbed into it and hence it cannot contribute much to the length of any 2-optimal tour (a \emph{sparse cell}), or many unperturbed origins are close to the cell (a \emph{heavy cell}). In the latter case, either the conditional densities of points perturbed into $C_i$ are small, hence any optimal tour inside $C_i$ has a large value by Lemma~\ref{lem:phi-lbOPT}, or we find another cell close to $C_i$ that has a very large contribution to the length of any tour.

To formalize this intuition, fix a cell $C_i$ and let $n_i$ be the expected number of points $X_j$ with $X_j\in C_i$. Assume for convenience that $a_1/\sigma$ and $(1+a_1)/\sigma$ are integer. We describe the position of a cube $C_i$ canonically by indices $\pos(C_i) \in \{-\frac{a_i}{\sigma},\dots, \frac{1+a_i}{\sigma}\}^d$. For two cells $C_i$ and $C_j$, we define their distance as $\distance(C_i,C_j) = \norm{\pos(C_i)-\pos(C_j)}_1$.
For $k\ge 0$, let $D_k$ denote all cells of distance $k$ to $C_i$ and let $n(D_k)$ denote the cardinality of unperturbed origins located in a cell in $D_k$. We call a perturbed point $X_\ell\in C_i$ with unperturbed origin $x_\ell\in C_j$, for some $C_j \in D_k$, a \emph{$k$-successful point}. Let $S_k$ denote the set of all $k$-successful points. Then $n_i = \sum_{k=0}^\infty \expected(|S_k|)$. 

Our first technical lemma shows that any cell $C_i$, having (in expectation) a large number $\mu$ of points perturbed into it from cells of distance at most $K$, contributes at least $\sigma \mu^{1-1/d} \exp(-O(K+1))$ to the length of the optimal tour. 

\begin{lemma}\label{lem:lb-k-successful}
Let $K\ge 0$ and define $S_{\le K}:= S_0\cup \cdots \cup S_{K}$ as the set of $k$-successful points for $k\le K$. Let $\mu := \expected(|S_{\le K}|)$. If $K=o(\log \mu)$, then with probability $1-\exp(-\Omega(\mu))$, we have
\[ L(E_{C_i}(T_\OPT)) \ge \frac{\sigma \mu^{1-1/d}}{\exp(O(K+1))}. \]
\end{lemma}
\begin{proof}%[of Lemma~\ref{lem:lb-k-successful}]
Note that by Lemma~\ref{lem:boundTSP}, $L(E_{C_i}(T_\OPT)) \ge \TSP(S_{\le K})-O(\sigma \cdot {|S_{\le K}|}^{\frac{d-2}{d-1}})$. Fix any realization of $S_{\le K}$, i.e., choice of unperturbed origins inside some cell in $D_0,\dots,D_K$ whose perturbed points fall into $C_i$. We can simulate the distribution of $\TSP(S_{\le K})$ (under this realization of $S_{\le K}$) by appealing to the one-step model.  Note that each point in $S_{\le K}$ is distributed as a Gaussian conditioned on containment in cell $C_j$. By rotational invariance of the Gaussian distribution, Lemma~\ref{lem:squaredens} is applicable and bounds the conditional density function of each point in $S_{\le K}$ by $\exp(K + (3/2)d)\sigma^{-d}$. By scaling, we obtain an instance in the unit cube with $N:=|S_{\le K}|$ points distributed according to density functions of maximum density $\exp(K+(3/2)d)$. Hence, by Lemma~\ref{lem:phi-lbOPT} we obtain that any tour has length $\Omega(N^{1-1/d}/\exp(K/d+3/2))$ on the scaled instance with probability $1-\exp(-\Omega(N))$. Scaling back to $C_i$, we obtain $\TSP(S_{\le K})\ge \Omega(\sigma N^{1-1/d}/\exp(K/d+3/2))$. Since by Chernoff bounds (Lemma~\ref{lem:chernoff}), $|S_{\le K}| = \Omega(\mu)$ with probability $1-\exp(-\Omega(\mu))$, we finally obtain, using \lemref{boundTSP},
\[L(E_{C_i}(T_\OPT)) \ge \Omega\left(\frac{\sigma \mu^{1-\frac{1}{d}}}{\exp(\frac{K}{d}+\frac{3}{2})}\right) - O(\sigma \cdot {\mu}^{\frac{d-2}{d-1}}) \ge \frac{\sigma \mu^{1-1/d}}{\exp(O(K+1))}, \]
with probability $1-\exp(-\Omega(\mu))$, where we used that $K=o(\log \mu)$.
\end{proof}

%\begin{proof}[Sketch]
%The claim follows from Lemma~\ref{lem:boundTSP} and by regarding $S_{\le K}$ as a $\phi$-perturbed instance. For this, Lemma~\ref{lem:squaredens} bounds the maximum density of the distributions and Lemma~\ref{lem:phi-lbOPT} bounds the optimal tour length from below.
%\end{proof}

The following simple technical lemma shows that with constant probability, a point is perturbed into the cell it originates in.

\begin{lemma}\label{lem:IntoOwnSquare}
Let $c\in Q_0 := [0,\sigma]^d$ and $Z \sim \Gauss(c,\sigma^2)$. Then $\probab(Z\in Q_0) \ge \frac{1}{(2\pi)^{d/2}}\exp(-\frac{d}{2})$.
\end{lemma}
\begin{proof}
Let $f(x)=\frac{1}{(2\pi)^{d/2}\sigma^d}\cdot \exp(-\frac{\norm{x-c}^2}{2\sigma^2})$ be the probability density function of $Z$. For all $x\in Q_0$, we have $\norm{x-c} \le \sqrt{d}\sigma$ and hence $f(x) \ge \frac{1}{(2\pi)^{d/2}\sigma^d}\cdot \exp(-\frac{d}{2})=: f_{\mathrm{min}}$. This yields 
\[\probab(Z\in Q_0) = \int_{Q_0} f(x) dx \ge \sigma^d  f_{\mathrm{min}} = \frac{1}{(2\pi)^{d/2}}\exp(-d/2).\qedhere\]
\end{proof}

We are set-up to formally show the classification of heavy cells. Recall that $M = \Theta(\sigma^{-d})$ denotes the number of cells $C_i$.

\begin{lemma}\label{lem:heavyCells}
Let $\alpha := M^\frac{d}{d-1}$, $k_1 :=\gamma \log \log (1/\sigma)$ and $k_2 := (1/\gamma') \sqrt{\log 1/\sigma}$ for sufficiently small constants $\gamma,\gamma'$. Then we can classify each cell $C_i$ with $n_i \ge \frac{n}{\alpha}$ into one of the following two types.
\begin{enumerate}[label=(T\arabic{*})]
\item\label{itm:Tclose} With probability $1-\exp(-\Omega(n^{1-\varepsilon}))$ for any constant $\varepsilon>0$, we have 
\[L(E_{C_i}(T)) \le O(\log 1/\sigma) L(E_{C_i}(T_\OPT)).\]
\item\label{itm:Tfar} There is some $C_j\in D_{k_1} \cup \cdots \cup D_{k_2}$ such that for any $f(1/\sigma)=\polylog(1/\sigma)$, we have 
\[L(E_{C_i}(T)) \le \frac{L(E_{C_j}(T_\OPT))}{f(1/\sigma)},\]
with probability $1-\exp(-\Omega(n^{1-\varepsilon}))$ for any constant $\varepsilon>0$.
\end{enumerate}
\end{lemma}
\begin{proof}
We start with some intuition. By Lemma~\ref{lem:2optworst}, we can bound $L(E_{C_i}(T)) = O(\sigma n_i^{1-1/d})$. If we have $\expected(|S_{\le k_1}|)=\Omega(n_i)$, then Lemma~\ref{lem:lb-k-successful} already proves $C_i$ to have type~\itemref{Tclose}. Otherwise, by tail bounds for the Gaussian distribution, we argue that some cell $C_j$ in distance at most $k_2$ contains at least $n_i \exp(\Omega((\log \log 1/\sigma)^2))$ unperturbed origins. These are sufficiently many to let $C_j$ contribute $f(1/\sigma)\sigma n_i^{1-1/d}$, for any $f(1/\sigma) = \polylog(1/\sigma)$, to the optimal tour length.

To make the intuition formal, note that all edges in $E_{C_i}(T)$ are contained in a cube of side length $3\sigma$ around $C_i$. By Chernoff bounds (Lemma~\ref{lem:chernoff}), at most $2n_i$ points are contained in $C_i$ with probability $1-\exp(-\Omega(n_i))$. Hence, Lemma~\ref{lem:2optworst} bounds 
\begin{equation}\label{eq:heavyUp}
L(E_{C_i}(T)) \le 3\sigma c_d (6n_i)^{1-1/d},
\end{equation}
with probability $1-\exp(-\Omega(n_i))$. 

\paragraph{Case 1: $\expected(S_{\le k_1}) > n_i/2$.} In this case, we may appeal to Lemma~\ref{lem:lb-k-successful} (since $k_1 = o(\log n_i)$) and obtain
\begin{equation}\label{eq:heavyLow1}
L(E_{C_i}(T_\OPT)) \ge  \frac{\sigma (|S_{\le k_1}|)^{1-1/d}}{\exp(O(k_1))} = \Omega\left( \frac{\sigma n_i^{1-1/d}}{\log(1/\sigma)}\right),
\end{equation}
with probability $1-\exp(-\Omega(n_i))$, since $k_1 = \gamma\log \log 1/\sigma$ and $\gamma$ can be chosen sufficiently small. By a union bound, \eqref{eq:heavyUp} and \eqref{eq:heavyLow1} hold with probability $1-\exp(-\Omega(n_i))=1-\exp(-\Omega(n^{1-\varepsilon}))$ for any constant $\varepsilon > 0$, proving that $C_i$ has type~\itemref{Tclose}.

\paragraph{Case 2: $\expected(S_{\le k_1}) \leq n_i/2$.} 
Every point in $C_i$ has an $\ell_1$-distance of at least $\sigma(\distance(C_i,C_j)-d)$ to every point in $C_j$. Thus, by Lemma~\ref{lem:chisquare}, we have
\begin{equation} \label{eq:ES_k}
 \expected(|S_k|) \le n(D_k) \probab\left(\norm{Z} \ge \frac{k-d}{\sqrt{d}}\sigma\right) \le n(D_k)\exp\left(-0.32\frac{(k-d)^2}{d}\right),
\end{equation}
for sufficiently large $k$.
Since $\alpha  = \mathrm{poly}(1/\sigma)$, we can choose a sufficiently small constant $\gamma'$ such that $k_2 = (1/\gamma') \sqrt{ \log 1/\sigma}$ satisfies $\exp(-0.32(k_2-d)^2/d) \le 1/(4\alpha)$. From $\sum_{k=0}^{\infty} n(D_k) = n$, we conclude
\[ \sum_{k=k_2+1}^\infty \expected(|S_k|) \le \sum_{k=k_2+1}^\infty n(D_k) \exp(-0.32(k-d)^2/d) \le \frac{n}{4\alpha} \le \frac{n_i}{4}.\]
Hence, we have
\[\sum_{k=k_1+1}^{k_2} \expected(|S_k|) = n_i - \expected(|S_{\le k_1}|)  - \sum_{k=k_2+1}^\infty \expected(|S_k|) \ge \frac{n_i}{4}.\]
By~\eqref{eq:ES_k}, it follows that
\[ N:= \sum_{k=k_1+1}^{k_2} n(D_k) \ge  \exp\left(0.32\frac{(k_1-d)^2}{d}\right) \sum_{k=k_1+1}^{k_2} \expected(|S_k|) = n_i \exp(\Omega((\log \log 1/\sigma)^2)) \]
unperturbed origins are situated in cells in distance $k_1 < k \le k_2$ from $C_i$. Note that there are at most $\sum_{k=k_1+1}^{k_2} |D_k| = O(k_2^d)=\polylog(1/\sigma)$ such cells and $\exp(\Omega((\log \log 1/\sigma)^2)=\omega(\log^c(1/\sigma))$ for any $c\in \mathbb{N}$. By pigeon hole principle, there is a cell $C_j\in D_{k_1}\cup \cdots \cup D_{k_2}$ with $\Omega(N/k_2^d) = n_i \exp(\Omega((\log \log 1/\sigma)^2))$ many unperturbed origins. 

Let $S'_0$ be the 0-successful points for cell $C_j$, i.e., the points with origin in $C_j$ that are perturbed into $C_j$. By Lemma~\ref{lem:IntoOwnSquare}, each unperturbed origin $x_\ell \in C_j$ has constant probability to be perturbed into $C_j$, i.e., $\probab(X_\ell \in C_j)=\Omega(1)$. Hence, $\expected(|S'_0|) = n_i \exp(\Omega((\log \log 1/\sigma)^2))$. Thus, Lemma~\ref{lem:lb-k-successful} bounds
\begin{equation} \label{eq:heavyLow2}
 L(E_{C_j}(T_\OPT)) \ge \frac{\sigma (\expected(|S'_0|))^{1-\frac{1}{d}}}{\exp(O(1))} = \sigma n_i^{1-\frac{1}{d}} \exp(\Omega((\log \log 1/\sigma)^2)),
\end{equation}
with probability $1-\exp(-\Omega(\expected(|S'_0|))) = 1-\exp(-\Omega(n_i))$. Since~\eqref{eq:heavyUp} and~\eqref{eq:heavyLow2} hold simultaneously with probability $1-\exp(-\Omega(n_i))=1-\exp(-\Omega(n^{1-\varepsilon}))$ for any constant $\varepsilon > 0$, this proves that $C_i$ has type~\itemref{Tfar}.

\end{proof}

\subsubsection{Total Length of 2-Optimal Tours}

With the analyses of the previous subsections, we can finally bound the total length of 2-optimal tours.
To bound the total length of short edges, consider first sparse cells~$C_i$, i.e., cells containing $n_i\le n/\alpha$ perturbed points in expectation (recall that $\alpha = M^\frac{d}{d-1}$, where $M = \Theta(\sigma^{-d})$ is the number of cells). For each such cell, the Chernoff bound (Lemma~\ref{lem:chernoff}) yields that with probability $1-\exp(-\Omega(n/\alpha))$, at most $2n/\alpha$ points are contained in~$C_i$, since each point is perturbed independently. By a union bound, no sparse cell contains more than $2n/\alpha$ points with probability at least $1-M\exp(-\Omega(n/\alpha))=1-\exp(-\Omega(n^{1-\varepsilon}))$ for any constant $\varepsilon>0$. In this event, Lemma~\ref{lem:2optworst} allows for bounding the contribution of sparse cells by
\begin{equation}\label{eq:sparse}
 \sum_{i: n_i \le n/\alpha} L(E_{C_i}(T)) \le M (3\sigma) c_d \left(\frac{6n}{\alpha}\right)^{1-\frac{1}{d}} = O\left(\frac{M\sigma  n^{1-\frac{1}{d}}}{\alpha^{1-\frac{1}{d}}}\right) = O(\sigma n^{1-\frac{1}{d}}).
\end{equation}
For bounding the length in the remaining cells (the heavy cells), let $\Tone := \{i \mid C_i \text{ has type~\itemref{Tclose}}\}$ and $\Ttwo := \{i \mid C_i \text{ has type \itemref{Tfar}}\}$. We observe the following: with probability at least $1-M\exp(-\Omega(n^{1-\varepsilon}))=1-\exp(-\Omega(n^{1-\varepsilon}))$, all type-\itemref{Tclose} cells $C_i$ satisfy $L(E_{C_i}(T)) = O(\log 1/\sigma) L(E_{C_i}(T_\OPT))$. Thus,
\begin{equation}\label{eq:heavyT1}
 \sum_{i \in \Tone} L(E_{C_i}(T)) \le  O(\log 1/\sigma) \cdot \left(\sum_{i \in \Tone} L(E_{C_i}(T_\OPT))\right) \le O(\log 1/\sigma) \cdot \OPT,
\end{equation}
where the last inequality follows from  $\sum_{i=1}^M L_{C_i}(T_\OPT) \le 2\cdot\OPT$, which holds since every edge in $\OPT$ (inside $C$) is counted at most twice on the left-hand side.

Let $A: \Ttwo \to \{1,\dots,M\}$ be any function that assigns to each 
cell $C_i$ of type-\itemref{Tfar} a corresponding cell $C_{A(i)} \in D_{k_1}\cup \cdots \cup D_{k_2}$ satisfying the condition~\itemref{Tfar}. We say that $C_i$ \emph{charges} $C_{A(i)}$. We can choose any $f(1/\sigma)=\polylog(1/\sigma)$ and have with probability  at least $1-M\exp(-\Omega(n^{1-\varepsilon}))=1-\exp(-\Omega(n^{1-\varepsilon}))$ that $L(E_{C_i}(T)) \le \frac{L(E_{C_{A(i)}}(T_\OPT))}{f(1/\sigma)}$ for all $i\in \Ttwo$. Assume that this event occurs. Since every cell $C_i$ can only be charged by cells in distance $k_1 < k \le k_2$, each cell can only be charged $\sum_{k=k_1+1}^{k_2} |D_k| = O(k_2^d)$ times. Hence,
\[ \sum_{i \in \Ttwo} L(E_{C_{A(i)}}(T_\OPT)) \le O(k_2^d) \sum_{i=1}^M  L(E_{C_{i}}(T_\OPT)) = O(k_2^d) \OPT. \]
Since $k_2^d=\polylog(1/\sigma)$, choosing $f(1/\sigma)=\polylog(1/\sigma)$ sufficiently large yields
\begin{equation}\label{eq:heavyT2}
 \sum_{i \in \Ttwo} L(E_{C_i}(T)) \le  \sum_{i \in \Ttwo} \frac{L(E_{C_{A(i)}}(T_\OPT))}{f(1/\sigma)} \le \frac{O(k_2^d)\OPT}{f(1/\sigma)} = O(\OPT).
\end{equation}

\begin{proof}[Proof of Theorem~\ref{thm:upper}]
By a union bound, we can bound by $1-\exp(-\Omega(n^{1/2-\varepsilon}))$, for any constant $\varepsilon>0$, the probability that (i) $\OPT=\Omega(\sigma n^{1-1/d})$ (by Lemma~\ref{lem:sigma-lbOPT}), (ii) all edges outside $C$ contribute $O(\sigma n^{1-1/d})=O(\OPT)$ (by Lemma~\ref{lem:outliers}), (iii) all sparse cells contribute $O(\sigma n^{1-1/d})=O(\OPT)$ (by \eqref{eq:sparse}), (iv) the type-\itemref{Tclose} cells~$C_i$ induce a cost of $O(\log 1/\sigma) \OPT$ (by \eqref{eq:heavyT1}), and (v) the type-\itemref{Tfar} cells induce a cost of $O(\OPT)$ (by~\eqref{eq:heavyT2}). Since the remaining edges are long edges and contribute only $O(\log(1/\sigma) \cdot \OPT)$, we obtain that every 2-optimal tour has a length of at most $O(\log 1/\sigma) \OPT$ with probability $1-\exp(-\Omega(n^{1/2-\varepsilon}))$. 

Since a 2-optimal tour always constitutes a $O(\log n)$-approximation to the optimal tour length by Lemma~\ref{lem:logworstcase}, we also obtain that the expected cost of the worst 2-optimal tour is bounded by 
\[
  O(\log 1/\sigma)\cdot \OPT + \exp(-\Omega(n^{1/2-\varepsilon})) \cdot O(\log n) \cdot \OPT = O(\log 1/\sigma) \cdot \OPT. \qedhere\] 
\end{proof}

\subsection{Lower Bound on the Approximation Ratio}
\label{sec:lb}

We complement our upper bound on the approximation performance by the following lower bound: for $\sigma = O(1/\sqrt{n})$, the worst-case lower bound is robust against perturbations. For this, we face the technical difficulty that in general, a single outlier might destroy the 2-optimality of a desired long tour, potentially cascading into a series of 2-Opt iterations that result in a substantially different or even optimal tour. 

\begin{theorem}\label{thm:2optlb}
Let $\sigma = O(1/\sqrt{n})$. For infinitely many $n$, there is an instance~$X$ of points in $\real^2$ perturbed by normally distributed noise of standard deviation $\sigma$ such that with probability $1-O(n^{-s})$ for any constant $s>0$, we have $\TwoOPT(X) = \Omega(\log n/\log \log n) \cdot \TSP(X)$. This also yields
\[ \expected\left(\frac{\TwoOPT(X)}{\TSP(X)}\right) = \Omega\left(\frac{\log n}{\log \log n}\right).\]
\end{theorem}

We remark that our result transfers naturally to the one-step model with $\phi = \Omega(n)$ and interestingly, holds \emph{with probability 1} over such random perturbations.

\begin{figure}
\includegraphics[width=\textwidth]{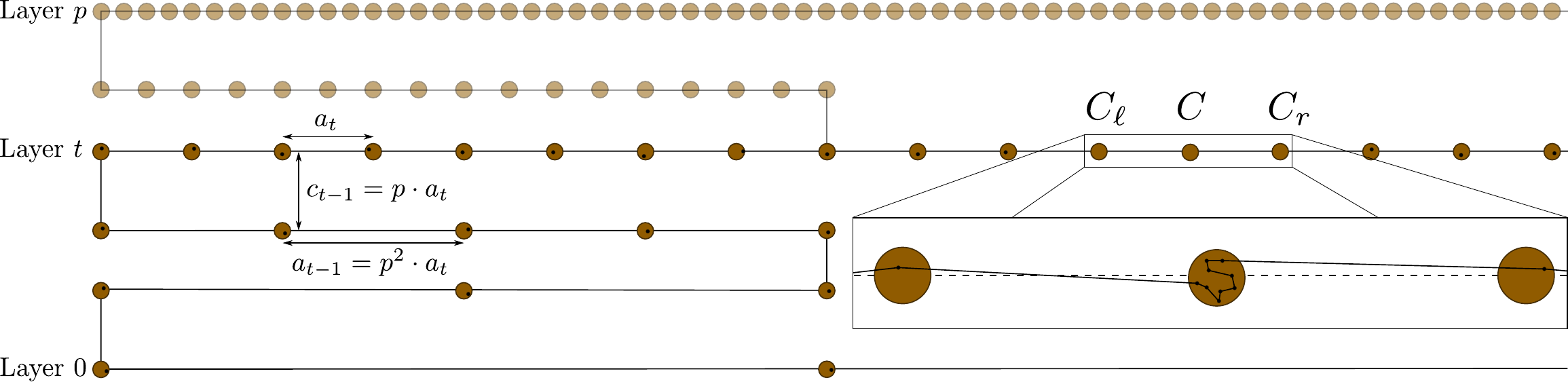}
\caption{Parts $V_1$ and $V_3$ of the lower bound instance. Each point is contained in a corresponding small container (depicted as brown circle) with high probability. The black lines indicate the constructed 2-optimal tour, which runs analogously on $V_2$.}
\label{fig:CKTconstruction}
\end{figure}

%%BM: TODO - remove the following sentence
%Furthermore, even when we initialize the tour using the \emph{nearest neighbor heuristic}, 
 %2-Opt might, with probability $O(1)$, return a 2-optimal tour of length $\Omega(\log n/\log \log n)\cdot \TSP(X)$ on perturbed inputs. We omit the necessary modifications to our construction here; they will be subject of future work.
%%For space reasons, the necessary changes to the construction below are deferred to a full version of this article.

\paragraph{Proof of Theorem~\ref{thm:2optlb}.}
 We alter the construction of Chandra et al.~\cite{ChandraEA:OldOpt:1999} to strengthen it against Gaussian perturbations with standard deviation $\sigma = O(1/\sqrt{n})$ (see Figure~\ref{fig:CKTconstruction}). Let $p\ge 3$ be an odd integer and $P := 3p^{2p}$. The original instance of~\cite{ChandraEA:OldOpt:1999} is a subset of the $(P\times P)$-grid, which we embed into $[0,1]^2$ by scaling by $1/P$, and consists of three parts $V_1$, $V_2$ and $V_3$. The vertices in $V_1$ are partitioned into the layers $L_0,\dots,L_p$. Layer $i$ consists of $p^{2i} + 1$ equidistant vertices, each of which has a vertical distance of $c_i = p^{2p-2i-1}/P$ to the point above it in Layer $i+1$ and a horizontal distance of $a_i = p^{2p-2i}/P$ to the nearest neighbor(s) in the same layer. The set $V_2$ is a copy of $V_1$ shifted to the right by a distance of $2/3$. The remaining part $V_3$ consists of a copy of Layer $p$ of $V_1$ shifted to the right by $1/3$ to connect $V_1$ and $V_2$ by a path of points. We regard $L_i$ as the set of Layer-$i$ points in $V_1\cup V_2 \cup V_3$.

As in the original construction, we will construct an instance of $n=\Theta(p^{2p})$ points, which implies $p=\Theta(\log n/\log \log n)$. Let $0\le t \le p$ be the largest odd integer such that $p^{2t+1} \le (3\sigma)^{-1}$. In our construction, we drop all Layers $t+1, \dots, p$ in both $V_1$ and $V_2$, as well as Layer $p$ in $V_3$. Instead, we connect $V_1$ and $V_2$ already in Layer $t$ by an altered copy 
 of Layer $t$ of $V_1$ shifted to the right by $1/3$. Let $C$ be an arbitrary point of our construction, for convenience we will use the central point of Layer $t$ in $V_3$. We introduce $p^{2p}-1$ additional copies of this point $C$. These surplus points serve as a ``padding'' of the instance to ensure $n=\Theta(p^{2p})$. Note that the resulting instance has $t+1$ layers $L_0, \dots, L_t$. 
We choose $t$ such that the magnitude of perturbation is negligible compared to the pairwise distances of all non-padding points. Furthermore, the restriction on~$\sigma$ ensures that incorporating the padding points increases the optimal tour length only by a constant.

\begin{lemma}\label{lem:optOnInstance}
With probability $1-O(n^{-s})$ for any constant $s>0$, the optimal tour has length $O(1)$.
\end{lemma}
\begin{proof}%[of Lemma~\ref{lem:optOnInstance}]
Let $n$ be the number of points in the constructed instance. Note that $\overline{X} = (x_1,\dots,x_n)$ consists of (i) a subset $\overline{X}^\orig$ of the instance of Chandra et al.~\cite{ChandraEA:OldOpt:1999}, plus (ii) an additional copy $\overline{X}^t$ of Layer $t$ and (iii) the padding points $\overline{X}^\pad$ in $V_3$. Denote the number of points in $\overline{X}^\orig \cup \overline{X}^t$ by $n'$. We have
\[ n' = p^{2t} + 2\left(\sum_{i=0}^t p^{2i} + 1\right) \le p^{2t} + 2(1-p^{-2})^{-1}p^{2t} + 2t = O(\sigma^{-1}/p),  \]
by choice of $t$. Hence $n= (p^{2p}-1) + n' = \Theta(p^{2p})$.  It is easy to see~\cite{ChandraEA:OldOpt:1999} that the original instance of Chandra et al.~has a minimum spanning tree of length
$\MST(\overline{X}^\orig) \le 9p^{2p}/P$. (This is achieved by the spanning tree that includes, for each Layer-$i$ vertex with $0\le i< p$, the vertical edge to the point above it, and each edge between consecutive points on Layer $p$.)
% BM removed \todo{redo this in the appendix?}.
Clearly,
\[
\MST(\overline{X}^\orig \cup \overline{X}^t) \le \MST(\overline{X}^\orig) + \MST(\overline{X}^t) + 2a_t \le 9p^{2p}/P + p^{2p}/P + 2a_t= O(1).
\]
Consider the perturbed instance $X \leftarrow \pert_\sigma(\overline X)$. 
Note that for every constant $s>0$, we have $p\sigma \ge 3\sigma \sqrt{d(s+1)\ln n}$ for sufficiently large $n$. Thus for each $1\le i \le n$, the Gaussian noise $Z_i \sim \Gauss(0,\sigma)$ satisfies $\norm{Z_i} \le p\sigma$ with probability at least $1-O(n^{-s+1})$ by Lemma~\ref{lem:chisquare}. By a union bound, we have $\sum_{i=1}^{n'} \norm{Z_i} \le O(n'p\sigma)=O(1)$ with probability at least $1-O(n^{-s})$. In this case, by the triangle inequality, the fact that $\TSP(Y) \le 2\cdot \MST(Y)$ for all point sets $Y$ and since only a constant number of edges connects the three parts, we obtain
\begin{align*}
\TSP(X_1,\dots,X_n) & \le  2\cdot \MST(\overline{X}^\orig \cup \overline{X}^t) + 2\left(\sum_{i=1}^{n'} \norm{Z_i}\right) + \TSP(X^\pad) + O(1)\\
 & \le  \TSP(X^\pad) + O(1).
\end{align*}
Note that we may translate and scale $\overline{X}^\pad$ to be contained in $[0,\sigma]^d$, by which $\TSP(X^\pad)$ may be regarded as the optimal tour length on an instance of $p^{2p} = \Theta(n)$ points in $[0,1]^d$ perturbed by Gaussians with standard deviation~1. By Lemma~\ref{lem:2optSigma1}, any 2-optimal tour and hence also the optimal tour on the scaled instance has length $O(\sqrt{n})$ with probability $1-\exp(-\Omega(\sqrt{n}))$. Scaling back to the original instance, we obtain $\TSP(X^\pad) = O(\sqrt{n}\sigma) = O(1)$ with probability $1-\exp(-\Omega(\sqrt{n}))$. This yields the result by a union bound.
\end{proof}

We find a long 2-optimal tour on all non-padding points analogously to the original construction by taking a shortcut of the original 2-optimal tour, which connects $V_1$ and $V_2$ already in Layer $t$ (see Figure~\ref{fig:CKTconstruction}). 

Consider the padding points, which are yet to be connected. Let $C_\ell$ denote the nearest point in Layer~$t$ of $V_3$ that is to the left of $C$. Symmetrically, $C_r$ is the nearest point to the right of $C$. Let $T^p$ be any 2-optimal path from $C_\ell$ to $C_r$ that passes through all the padding points (including $C$). We replace the edges $(C_\ell, C)$ and $(C,C_r)$ by the path $T^p$, completing the construction of our tour $T$.

\begin{lemma}\label{lem:prob2opt}
Let $s>0$ be arbitrary. With probability $1-O(n^{-s})$, $T$ is 2-optimal and has a length of $\Omega(\log n / \log \log n)$.%$\Omega\left(\frac{\log n}{\log \log n}\right)$.
\end{lemma}

Note that given \lemref{prob2opt}, \thmref{2optlb} follows directly using \lemref{optOnInstance}. The (rather technical) proof of \lemref{prob2opt} hence concludes our lower bound.

\paragraph{Probability of 2-optimality.}

To account for the perturbation in the analysis, we define a safe region for every point. More formally, let $x_j$ be any unperturbed origin. We define its \emph{container} $B_j$ as the circle centered at $x_j$ with radius $\beta := a_t/8 = p^{2p-2t}/(8P) \ge \sigma p/8$. Very likely, all perturbed points lie in their containers.

%\begin{restatable}{lemma}{TisTwoOpt}
\begin{lemma}
\label{lem:Tis2opt}
For sufficiently large $p$, the tour $T$ constructed as described in Section~\ref{sec:lb} is 2-optimal, provided that all points $X_j$ lie in their corresponding containers~$B_j$.
\end{lemma}
%\end{restatable}

We first show that this lemma implies Lemma~\ref{lem:prob2opt}.

\begin{proof}[Proof of Lemma~\ref{lem:prob2opt}]
Let $Z \sim \Gauss(0,\sigma^2)$, and let $s>0$ be arbitrary. By $\beta \ge \sigma p/8 = \Omega(\sigma \log(n)/\log\log n) = \omega(\sigma\sqrt{\log n})$, we have $\beta \ge 3\sigma\sqrt{d(s+1)\ln n}$ for sufficiently large~$n$. By definition of the containers, Lemma~\ref{lem:chisquare} yields that for any point $X_j$ and sufficiently large $n$,
\[ \probab(X_j \notin B_j) \le \probab(\norm{Z}\ge \beta) \le \probab(\norm{Z} \ge \sigma 3\sqrt{d (s+1) \ln n}) \le n^{-(s+1)}. \]
By a union bound, we conclude that with probability $1-n^{-s}$, all points are contained in their corresponding containers and hence, by the previous lemma, $T$ is 2-optimal.

Recall that $t$ is the largest odd integer satisfying $p^{2t+1} \le (3\sigma)^{-1}$. Since $\sigma^{-1} = \Omega(\sqrt{n})$, this implies $t \ge \frac{p-1}{2}-1$. Observe that $T$ visits $t = \Omega(p)$ many layers and crosses a horizontal distance of $2/3$ in each of them. Hence, it has a length of at least $\Omega(p) = \Omega(\log n/\log \log n)$.
\end{proof}

%\TisTwoOpt*
In the remainder of this section, we prove Lemma~\ref{lem:Tis2opt}, i.e., show that the constructed tour is 2-optimal, provided all points stay inside their respective containers. Clearly, it suffices to show for any pair of edges $(u,v)$ and $(w,z)$ in the tour, the corresponding \emph{2-change}, i.e., replacing these edges by $(u,w)$ and $(v,z)$ does not reduce the tour length, i.e., $d(u,w)+d(v,z) \ge d(u,v)+d(w,z)$.  We first state the technical lemmas capturing the ideas behind the construction.

The first lemma treats pairs of horizontal edges and establishes how large their vertical distance must be in order to make swapping these edges increase the length of the tour. It is a generalization of a similar lemma of Chandra et al.~\cite{ChandraEA:OldOpt:1999} to a perturbation setting, in which points are placed arbitrarily into small containers.

Note that in what follows, for a point $p\in \real^2$, we let $p_x$ denote its $x$-coordinate and $p_y$ its $y$-coordinate. Furthermore, for any points $p,q\in \real^2$, we let $d_x(p,q):=|p_x-q_x|$ and $d_y(p,q):=|p_y-q_y|$ denote their horizontal and vertical distance, respectively.

\begin{lemma}\label{lem:horiz}
Let $\overline{pq}$ and $\overline{rs}$ be horizontal line segments in the Euclidean plane with $p_x < q_x$ and $r_x < s_x$. Let $B_p$, $B_q$, $B_r$ and $B_s$ be circles of radius $\beta$ with centers $p$, $q$, $r$ and~$s$, respectively. If $d(r,s) \ge d(p,q) + 4\beta$ and the vertical distance $v:=d_y(p,r) = d_y(q,s)$ between $\overline{pq}$ and $\overline{rs}$ is at least
\[ \sqrt{d(p,q)d(r,s) + 4\beta d(r,s)} + 2\beta, \] 
then, for all $\tilde{p}\in B_p, \tilde{q} \in B_q, \tilde{r}\in B_r,\tilde{s}\in B_s$, we have
\[ d(\tilde{p},\tilde{r}) + d(\tilde{q},\tilde{s}) \ge d(\tilde{p},\tilde{q}) + d(\tilde{r},\tilde{s}). \]
\end{lemma}
\begin{proof}
Note that $d_y(\tilde{p},\tilde{r}), d_y(\tilde{q},\tilde{s}) \ge v - 2\beta$. Furthermore, we have that 
\begin{eqnarray*}
d(r,s)-2\beta \le \tilde{s}_x - \tilde{r}_x & = & (\tilde{p}_x - \tilde{r}_x) + (\tilde{q}_x-\tilde{p}_x) + (\tilde{s}_x - \tilde{q}_x) \\
& \le & (\tilde{p}_x - \tilde{r}_x) + d(p,q) + 2\beta + (\tilde{s}_x - \tilde{q}_x),
\end{eqnarray*}
and hence
\[ (\tilde{p}_x - \tilde{r}_x) + (\tilde{s}_x -\tilde{q}_x) \ge   d(r,s) - d(p,q) - 4\beta ,\]
where the right-hand side expression is at least $0$, since $d(r,s)\ge d(p,q) + 4\beta$ by assumption. Let $L := \tilde{p}_x - \tilde{r}_x$ and $R:=\tilde{s}_x -\tilde{q}_x$, then it is straight-forward to verify that the expression
\begin{eqnarray}
d(\tilde{p},\tilde{r}) + d(\tilde{q},\tilde{s}) &  \ge & \sqrt{(v-2\beta)^2 + L^2} + \sqrt{(v-2\beta)^2 + R^2} \label{eq:cross},
\end{eqnarray}
subject to $L+R \ge d(r,s)-d(p,q) - 4\beta$ is minimized when $L=R=\frac{d(r,s)-d(p,q)}{2}-2\beta$.
%The right-hand side of the previous inequality is minimized when $x = \frac{d(r,s)-d(p,q)}{2}-2\beta$. %To see this\footnote{CKT do exactly this, but only for $x\in[0,B]$}, consider the function $f(x) := \sqrt{A+x^2} + \sqrt{A+(B-x)^2}$. We compute
%\begin{eqnarray*}
%f'(x) & = & \frac{x}{\sqrt{A+x^2}} - \frac{B-x}{\sqrt{A+(B-x)^2}}. \\
%\end{eqnarray*}
%For $0\le x \le B$, this yields
%\[ f'(x) = \frac{1}{\sqrt{1+\frac{A}{x^2}}} - \frac{1}{\sqrt{1+\frac{A}{(B-x)^2}}}, \]
%which fulfills $f'(x) >0$ if $B/2 < x < B$, since then $\frac{A}{x^2} \le \frac{A}{(B-x)^2}$. Similarly, for $0<x< B/2$, we have $f'(x) < 0$.

%For $x<0$, both summands $x/\sqrt{A+x^2}$ and $-(B-x)/\sqrt{A+(B-x)^2}$ are negative, hence $f'(x) < 0$. Similarly, for $x > B$, both summands are positive and $f'(x) > 0$. Thus, $x=B/2$ minimizes $f(x)$, since $f'(x) < 0$ for all $x<B/2$ and $f'(x) > 0$ for all $x>B/2$.

Hence, we can bound \eqref{eq:cross} by
\begin{eqnarray*}
\lefteqn{d(\tilde{p},\tilde{r}) + d(\tilde{q},\tilde{s})} \\
 & \ge & 2\sqrt{(v-2\beta)^2 + \left(\frac{d(r,s)-d(p,q)}{2}-2\beta\right)^2} \\
& \ge & 2\sqrt{d(p,q)d(r,s) + 4\beta d(r,s) + \left(\frac{d(r,s)-d(p,q)}{2}-2\beta\right)^2} \\
& = & 2\sqrt{\left(\frac{d(r,s)-d(p,q)}{2}\right)^2+d(p,q)d(r,s) + 4\beta d(r,s) - 2\beta(d(r,s)-d(p,q)) + (2\beta)^2} \\
& = & 2\sqrt{\left(\frac{d(r,s)+d(p,q)}{2}\right)^2 + 2\beta (d(r,s) + d(p,q)) + (2\beta)^2} \\
& = & 2 \left(\frac{d(r,s)+d(p,q)}{2} + 2\beta\right) = d(r,s) + d(p,q) + 4\beta \ge d(\tilde{p},\tilde{q}) + d(\tilde{r},\tilde{s}),
\end{eqnarray*}
where the third line follows from our assumption on $v$.
\end{proof}

The following very basic lemma shows that a sequence of edges that share roughly the same direction will always be 2-optimal.

\begin{lemma}\label{lem:45angle}
Let $p_1,p_2,p_3$ and $p_4$ be a sequence of points in $[0,1]^2$ such that all connecting segments $p_{i+1}-p_i$ fulfill $|(p_{i+1} - p_i)_y| \le (p_{i+1} - p_i)_x$. Then,
\[ d(p_1,p_3) + d(p_2,p_4) \ge d(p_1,p_2) + d(p_3, p_4).\]
\end{lemma}
\begin{proof}
For any point $p$, let $C_p$ denote the cone $C_p := \{q \mid |(q-p)_y| \le (q-p)_x\}$. Let $\Delta := p_2 - p_1$, then by assumption, we have $p_2 \in C_{p_1}$ and thus $|\Delta_y| \le \Delta_x$. Let us assume that $0 \le \Delta_y \le \Delta_x$ (the other case is symmetric). Since by assumption, $p_3 \in C_{p_2}$, we have for $\Delta' := p_3 - p_1$ that $\Delta'_x = \Delta_x + \delta_x$ and $\Delta'_y = \Delta'_y + \delta_y$ for some $\delta_x > 0$ and $\delta_y$ with $|\delta_y| < \delta_x$. If $\delta_x \ge \Delta_y$, the claim is immediate from $d(p_1,p_3) \ge \Delta_x + \delta_x \ge \Delta_x+\Delta_y\ge d(p_1,p_2)$. Otherwise, for $\delta_x <\Delta_y$, we obtain
\begin{eqnarray*}
d(p_1, p_3) & = & \sqrt{(\Delta_x+\delta_x)^2 + (\Delta_y+\delta_y)^2} \\
& \ge & \sqrt{(\Delta_x+\delta_x)^2 + (\Delta_y-\delta_x)^2}\\
& \ge & \sqrt{\Delta^2_x + \Delta_y^2 + 2\delta_x(\Delta_x - \Delta_y)} \ge  \sqrt{\Delta^2_x + \Delta_y^2} = d(p_1, p_2).
\end{eqnarray*}
By an analogous computation, $d(p_2,p_4)\ge d(p_3, p_4)$ follows and hence the claim.
\end{proof}

%\begin{lemma}[Long edges in layer $p$.]
%Consider any $w\in P_p$. Let $k\in \N$, $m=\sqrt{18\ln 3}\sigma$ and $z$ be the closest point in $P_p \cap C_w$. The probability that $d_x(x,y)> 3m + k/n$ is bounded by $(1/2+3^{-5.8})^k$.
%\end{lemma}
%\begin{proof}
%Note that the cone $C_w$ has a height of $2m$ at the horizontal distance $m$ from $q$. Hence for all $t\ge 2m$, the circle $M_t$ with radius $m$ and center $[w_x + t, 1]$ is completely contained in $C_w\cap [0,1]\times[1-m,1+m]$. Since the layer-$p$ points are points $(\ell/n,1)$ with $\ell \in \N$ and $0\le \ell \le 3p^{2p}$, we can find perturbed points $Y_1,Y_2,\dots$ with means $y_i$ having $(y_i)_x \le w_x + 2m + i/n$ such that the circles $M^i$ of radius $m$ and center $y_i$ fulfill $M^i \subseteq C_w \cap P_p$. By Lemma~\ref{lem:chisquare}, 
%\[\Pr[Y_i \notin M^i] \le \Pr[\norm{Y_i-y_i} \ge m] \le 3^{-5.8}.\]
%If $Y_i\in M^i$, by rotation-invariance of the Gaussian distribution, we conclude that with probability $1/2$, the \emph{successful event} $Y_i \in M^i \cap [0,1]\times[1,1+m] \subseteq C_w \cap P_p$ occurs. If it does, $d_x(w,Y_i)\le (y_i)_x+m - w_x \le 3m + i/n$. Thus, the probability that the horizontal distance to the closest point in $P_p\cap C_w$ is larger than $3m + k/n$ is bounded by the probability that for none of the points $Y_1, \dots, Y_k$, the good event occurs. This probability is bounded by $(1/2+3^{-5.8})^k$.
%\end{proof}

We can now prove Lemma~\ref{lem:Tis2opt}. Assume that all points are contained in their respective containers. We call an edge between $X_i$ and $X_j$ \emph{horizontal} (or \emph{vertical}) if the edge between $x_i$ and $x_j$ is horizontal (or vertical) and neither $x_i$ nor $x_j$ belong to the set of padding points.
In what follows, we will first consider horizontal-horizontal, horizontal-vertical and vertical-vertical edge pairs and then turn to pairs of edges for which at least one edge is adjacent to some padding point. Recall that $\beta$ is chosen such as to satisfy $a_t = 8\beta$. %To ease the upcoming case distinction, we first consider pairs of edges $(X_i,X_{i'})$ and $(X_j, X_{j'})$ with $X_i,X_{i'} \in V_1$ and $X_j, X_{j'}\in V_2$. By construction, $d(X_i,X_j) \ge p^{2p}-\beta_0$

%\paragraph{Vertical-vertical pair.} Let $(X_i,X_{i+1})$ and $(X_j,X_{j+1})$ be two vertical edges. We distinguish three cases.
%\begin{enumerate}
%\item $(x_i)_x = (x_{i+1})_x = (x_j)_x = (x_{j+1})_x$. Then, by construction, 
% \end{enumerate}

\paragraph{Horizontal-horizontal edge pair.} Let $(X_i,X_{i+1})$ and $(X_j,X_{j+1})$ be two horizontal edges. Horizontal edges $(X_i,X_{i+1})$ with $x_i,x_{i+1}\in L_k$ appear only if $k\le t$. We distinguish the following cases.
\begin{enumerate}
\item $x_i,x_{i+1},x_j, x_{j+1} \in L_k$: Both edges are in the same layer. Note that no 2-change swaps neighboring edges. Assume without loss of generality that $(x_i)_x < (x_{i+1})_x < (x_j)_x < (x_{j+1})_x$ (the other case is symmetric). Since $a_k \ge a_t =  8\beta$, we have that
\[d_y(X_i,X_{i+1}) \le 2\beta \le a_k - 2\beta \le d_x(X_i,X_{i+1}).\]
Similarly, $d_y(X_j,X_{j+1}) \le d_x(X_j,X_{j+1})$ and $d_y(X_{i+1},X_j) \le d_x(X_{i+1},X_j)$. This shows that Lemma~\ref{lem:45angle} is applicable to $X_i,X_{i+1},X_{j},X_{j+1}$, which yields that no 2-change can be profitable.
\item $x_i,x_{i+1} \in L_k$, and $x_j,x_{j+1}\in L_{k+1}$. By construction of $T$, the edges have opposite direction. Assume that $(x_i)_x < (x_{i+1})_x$ and hence $(x_j)_x > (x_{j+1})_x$ (the other case is symmetric). By construction $(x_{i+1})_x - (x_i)_x \ge a_k$. We have that $(X_{i+1})_x - (X_i)_x \ge a_k - 2\beta \ge a_t - 2\beta = 6\beta > 0$. The same reasoning shows that $(X_j)_x > (X_{j+1})_x$. Similarly, one can show that $p_y > q_y$ for all $p\in \{X_j, X_{j+1}\}$ and $q\in \{X_i, X_{i+1}\}$.  Hence the 2-change to $(X_i,X_j)$ and $(X_{i+1},X_{j+1})$ has a crossing, which by triangle inequality cannot be profitable.
\item $x_i,x_{i+1} \in L_k$, and $x_j,x_{j+1} \in L_{k+\ell}$ with $\ell \ge 2$ and $k+\ell \le t$. Either both edges have opposite directions, then the previous argument shows that a 2-change is not profitable. Otherwise, note that the first requirement of Lemma~\ref{lem:horiz}, $a_k \ge a_{k+\ell} + 4\beta$, is fulfilled. Also note that $\beta = \frac{a_t}{8} \le \frac{a_k}{8p^{2\ell}}$, since $k\le t-\ell$. We have
\begin{eqnarray*} 
\sqrt{d(x_i,x_{i+1})d(x_j,x_{j+1}) + 4\beta d(x_i,x_{i+1})} + 2\beta & = & \sqrt{a_k a_{k+\ell} + 4\beta a_k} + 2\beta \\
& \le & \sqrt{\frac{a_k^2}{p^{2\ell}}  + \frac{a^2_k}{2p^{2\ell}}} + \frac{a_k}{4p^{2\ell}} \\
& \le & \sqrt{\frac{3}{2}}\cdot \frac{a_k}{p^\ell} + \frac{a_k}{4p^{2\ell}} \\
& \le & \left(\sqrt{\frac{3}{2}}\cdot \frac{1}{p^{\ell-1}} + \frac{1}{4p^{2\ell-1}}\right)\frac{a_k}{p} \\
& \le & c_k \le \sum_{m=0}^{\ell-1} c_{k+m} = d_y(x_i,x_j),
\end{eqnarray*}
since for sufficiently large $p$, we have $\sqrt{3/2}/p^{\ell-1} + 1/(4p^{2\ell-1})\le 1$. Consequently, Lemma~\ref{lem:horiz} applies and shows that the 2-change does not yield an improvement.
\end{enumerate}

\paragraph{Horizontal-vertical edge pair.} Let $(X_i,X_{i+1})$ be a vertical edge and $(X_j, X_{j+1})$ be a horizontal edge. We assume that the vertical edge is in $V_1$, since the case $x_i, x_{i+1} \in V_2$ is symmetric. Exactly one of the following cases occurs.
\begin{enumerate}
\item $x_i \in L_k, x_{i+1} \in L_{k+1}$ and $x_j, x_{j+1} \in L_{k'}$ with $k' \in \{k, k+1\}$. The horizontal edge is in the same layer as one of the end points of the vertical edge. Clearly, $d(X_i, X_{i+1}) \le c_k + 2\beta$ and $d(X_j, X_{j+1}) \le a_{k'} + 2\beta$. Since a 2-change cannot swap neighboring edges, at least one horizontal segment lies between both edges. By construction of the tour, one of the edges $\{x_i, x_j\}$ and $\{x_{i+1}, x_{j+1}\}$ crosses a vertical distance of at least $c_k$ and the other a horizontal distance of at least $2a_{k'}$. Hence 
\[ d(X_i, X_j)+d(X_{i+1}, X_{j+1}) \ge 2a_{k'} + c_k - 4\beta \ge a_{k'} + c_k + 4\beta, \]
since $a_{k'} \ge a_t = 8\beta$.
\item $x_i \in L_k,x_{i+1}\in L_{k+1}$ and $x_j,x_{j+1} \in L_{k'}$ with $k' \notin \{k, k+1\}$. As in the previous case, $d(X_i, X_{i+1}) \le c_k + 2\beta$ and $d(X_j, X_{j+1}) \le a_{k'} + 2\beta$. Consider first the case that $k' < k$, then by construction of the tour, one of the edges $\{x_i,x_j\}$ and $\{x_{i+1}, x_{j+1}\}$ crosses a horizontal distance of at least $a_{k'}$ and the other edge crosses a vertical distance of at least  $c_{k'}$, yielding \[d(X_i,X_j) + d(X_{i+1},X_{j+1}) \ge a_{k'} + c_{k'} - 4\beta \ge a_{k'} + c_{k} + 4\beta, \] since $c_{k' }\ge c_{k-1} \ge c_k + 8\beta$. Otherwise, if $k'> k+1$, the edges ${x_i,x_j}$ crosses a vertical distance of at least $c_{k+1}+c_k$ and hence \[d(X_i,X_j) + d(X_{i+1},X_{j+1}) \ge  c_{k+1} + c_k - 2\beta \ge a_{k'} + c_{k} + 4\beta, \] since $c_{k+1} \ge a_{k+2} + 6\beta \ge a_{k'}  +6\beta$. Thus in both cases, a 2-change is not profitable. %$\sum_{\ell = k'}^{k} c_\ell$ if $\ell < k$. In either case, this distance is at least $c_k + c_{t-1}$. Hence
%\item $x_i \in L_k,x_{i+1}\in L_{k+1}$ and $x_j,x_{j+1} \in L_{k'}$ with $k' \notin \{k, k+1\}$. As in the previous case, $d(X_i, X_{i+1}) \le c_k + 2\beta$ and $d(X_j, X_{j+1}) \le a_{k'} + 2\beta$. By construction of the tour, one of the edges $\{x_i,x_{i+1}\}$ and $\{x_j, x_{j+1}\}$ crosses a horizontal distance of at least $a_{k'}$. The other edge crosses a vertical distance of either $\sum_{\ell = k}^{k'-1} c_\ell$ if $k' > k + 1 $ or $\sum_{\ell = k'}^{k} c_\ell$ if $\ell < k$. In either case, this distance is at least $c_k + c_{t-1}$. Hence
%\begin{eqnarray*}
% d(X_i, X_j) + d(X_{i+1},X_{j+1}) & \ge & c_k + a_{k'} + c_{t-1} - 4 \beta \\
%& \ge & c_k + a_{k'} + 4\beta \ge d(X_i, X_{i+1}) + d(X_j, X_{j+1}), \end{eqnarray*}
%since $c_{t-1} \ge pa_t \ge 8\beta$.
\end{enumerate}

\paragraph{Vertical-vertical edge pair.} Let $(X_i,X_{i+1})$ and $(X_j, X_{j+1})$ be vertical edges. 

\begin{enumerate}
\item $x_i \in L_k, x_{i+1} \in L_{k+1}$ and $x_j \in L_{k'}, x_{j+1} \in L_{k'+1}$ with $(x_i)_x = (x_j)_x$, i.e., the vertical edges are above each other. By swapping the $x$- and $y$-axis in Lemma~\ref{lem:45angle}, we can show that a 2-change is not profitable, since it is easy to see that $|(p-q)_x| \le (p-q)_y$ for all consecutive pairs $(p,q)$ in $(X_i, X_{i+1}, X_j, X_{j+1})$. 
\item $x_i \in L_k, x_{i+1} \in L_{k+1}$ and $x_j \in L_{k'}, x_{j+1} \in L_{k'+1}$ with $(x_i)_x \ne (x_{i'})_x$.
Clearly, $d(X_i, X_j) \ge a_0 - 2\beta$ and $d(X_{i+1}, X_{j+1}) \ge a_0 - 2\beta$, while $d(X_i, X_{i+1}) \le c_k +2\beta \le c_0 + 2\beta$ and $d(X_j, X_{j+1})\le c_{k'} + 2\beta \le c_0 + 2\beta$. Hence a 2-change is not profitable, since $a_0 \ge 8\beta + c_0$.
\end{enumerate}

\paragraph{Padding points.}

Since we assumed for convenience that the padding points are placed at the central vertex $C$ of Layer $t$ in $V_3$, only the edges with at least one endpoint in $V_3$ are relevant candidates for the treatment of padding points. This is because all other edges have both endpoints at a distance of $1/6$ to the padding points, which can never be accounted for by its edge length, since all edges except in Layer 0 are much shorter than $1/3$. Separately, the Layer-0 edges can be handled easily as well: an edge $\{X_i, X_{i'}\}$ with $x_i = x_{i'}\in \overline{X}^\pad$ is a horizontal edge, hence the pair $(X_i,X_{i'})$ and a Layer-0 edge trigger the corresponding case of horizontal-horizontal edge pairs with even smaller edge length of the edge $(X_i,X_{i'})$ in Layer $t$.

It remains to handle the following cases, where we regard $C$ as a padding point, i.e., $C\in \overline{X}^\pad$, not as a Layer-$t$ point. 

\begin{enumerate}
\item $x_i, x_{i'} \in \overline{X}^\pad$, and $x_j, x_{j'} \in L_t$. Clearly, $d(X_j,X_{j'}) \le a_t + 2\beta$ and $d(X_j, X_{j'}) \le 2\beta$. Furthermore, at least one of $\{x_j, x_{j'}\}$ has a horizontal distance of at least $2a_t$ to $x_i=x_{i'}$. Hence, 
\[d(X_i, X_j) + d(X_{i'}, X_{j'})\ge 2a_t - 2\beta \ge a_t + 4\beta \ge d(X_j, X_{j'}) + d(X_i, X_{i'})   \]
\item $x_i\in \{C_\ell,C_r\}, x_{i'} \in \overline{X}^\pad$ and $x_j, x_{j'} \in L_t$. These edge pairs are exactly as regular pairs of Layer-$t$ edges and the corresponding case of horizontal-horizontal edge pairs applies.
\item $x_i, x_{i'}, x_j, x_{j'} \in \overline{X}^\pad \cup \{C_r, C_\ell\}$. All such edges are 2-optimal by construction, since a \emph{2-optimal path} from $C_\ell$ to $C_r$ passing by all padding points was used.
\end{enumerate}

This concludes the case analysis and thus the proof of \lemref{Tis2opt}.

%\subsubsection{Proof of Theorem~\ref{thm:2optlb}}
%\label{sec:thmproofsigma}
%
%\begin{proof}[Proof of Theorem~\ref{thm:2optlb}]
%For all sufficiently large odd numbers $p$, we have described an instance with origins $(x_1,\dots,x_n)$ such that, by Lemma~\ref{lem:prob2opt}, with probability $1-O(n^{-s})$ for any constant $s>0$, there is a 2-optimal tour of length $\Omega\left(\frac{\log n}{\log \log n}\right)$ on $X\sim \pert_\sigma(x)$. However, Lemma~\ref{lem:optOnInstance} shows that with probability $1-O(n^{-s})$, the optimal tour has a length of $O(1)$. This yields the first claim and immediately proves that the expected approximation ratio is lower bounded by $\Omega(\log n/\log \log n)$.%Additionally, 
%%\[ \Eop_{X \sim\pert_\sigma(X)}\left[\frac{\TwoOPT(X)}{\TSP(X)}\right] \ge (1 - o(1)) \cdot \Omega\left(\frac{\log n}{\log \log n}\right) + o(1)\cdot 1. \]
%\end{proof}

\section{Concluding Remarks}
\label{sec:concl}

\paragraph{Running-time.}
Our approach for Euclidean distances does not work for $d=2$ and $d=3$. However, we can use the bound of Englert et al.~\cite{EnglertEA:2Opt:2014}
for Euclidean distances, which yields a bound polynomial in $n$ and $1/\sigma$ for $d \in \{2,3\}$.

In the same way as Englert et al.~\cite{EnglertEA:2Opt:2014}, we can slightly
improve the smoothed number of iterations by using an insertion heuristic to choose the initial tour. We save
a factor of $n^{1/d}$ for Manhattan and Euclidean distances
and a factor of $n^{2/d}$ for squared Euclidean distances.
The reason is that there always exist tours of length $O(\maxx n^{1 - \frac 1d})$
for $n$ points in $[-\maxx, \maxx]^d$ for Euclidean and Manhattan distances
and of length $O(\maxx^2 n^{1 - \frac 2d})$ for squared Euclidean
distances for $d \geq 2$~\cite{Yukich:ProbEuclidean:1998} (the constants in these upper bounds depend on $d$).
Taking into account also that, because Gaussians have light tails, only few points are far away from the hypercube $[0,1]^d$ after perturbation,
one might get an even better bound. However, we did not take these improvements into account in our analysis to keep the paper concise.

Of course, even our improved bounds do not fully explain the 
linear number of iterations observed in experiments.
However, we believe that new approaches, beyond analyzing the smallest improvement, are needed in order to further improve the smoothed bounds on the running-time.

\paragraph{Approximation ratio.}

We have proved an upper bound of $O(\log 1/\sigma)$ for the smoothed approximation
ratio of 2-Opt. Furthermore, we have proved that the lower bound of Chandra et al.~\cite{ChandraEA:OldOpt:1999}
remains robust even for $\sigma = O(1/\sqrt{n})$. We leave as an open problem to generalize our upper bounds to the one-step model to improve the current
bound of $O(\sqrt[d]{\phi})$~\cite{EnglertEA:2Opt:2014}, but conjecture
that this might be difficult, because of the special structure that Gaussian distributions provide.

While our bound significantly improves the previously known bound for the smoothed approximation ratio of 2-Opt, we readily admit
that it still does not explain the performance observed in practice.
A possible explanation is that when the initial tour is not picked by an adversary or the nearest neighbor heuristic, but using a construction heuristic such as the spanning tree heuristic or an insertion heuristic, an approximation factor of $2$ is guaranteed even before 2-Opt has begun to improve the tour~\cite{RosenkrantzEA:AnalysisHeuristicsTSP:1977}. We chose to compare the worst local optimum to the global optimum in order, as this is arguably the
simplest of all technically difficult possibilities.

However, a smoothed analysis of the approximation ratio of 2-Opt initialized with a good heuristic might be difficult: even in the average case,
it is only known that the length of an optimal TSP is concentrated around $\gamma_d \cdot n^{\frac{d-1}d}$ for some constant $\gamma_d > 0$.
But the precise value of $\gamma_d$ is unknown~\cite{Yukich:ProbEuclidean:1998}. Since experiments suggest that 2-Opt even with good initialization does not
achieve an approximation ratio of $1+o(1)$~\cite{JohnsonMcGeoch:CaseTSP:1997,JohnsonMcGeoch:ExperimentalSTSP:2002}, one has to deal with the precise
constants, which seems challenging.

Finally, we conjecture that many examples for showing lower bounds for the approximation ratio of concrete algorithms
for Euclidean optimization such as the TSP remain stable under perturbation for $\sigma = O(1/\sqrt n)$.
The question remains whether such small values of~$\sigma$, although they often suffice to prove polynomial smoothed running time,
are essential to explain practical approximation ratios or if already slower decreasing $\sigma$ provide a sufficient explanation.

%\bibliographystyle{plain}
%\bibliography{abbrev,bodo,papers,chapters,books}

\end{document}